\documentclass[manuscript,screen,nonacm]{acmart}

\usepackage{amsfonts}
\usepackage{verbatim}
\usepackage{amsfonts}
\usepackage[english]{babel}
\usepackage{color}
\usepackage{xspace}

\usepackage{tikz}
\usepackage{hyperref}
\usepackage{float}
\usepackage{enumerate}
\usepackage{enumitem}
\usepackage{subcaption}
\usepackage{scalerel}
\usepackage{booktabs}
\usepackage{array}

\newcommand{\dep}{\ensuremath{\mathop{=\!}}\xspace}
\newcommand{\pci}[3]{{#2 \perp\!\!\!\perp_{#1} #3}}

\newcommand{\var}{\ensuremath{\mathsf{var}}\xspace}

\newcommand{\cl}{\ensuremath{\mathsf{cl}}\xspace}
\newcommand{\desclist}{\ensuremath{\mathsf{desclist}}\xspace}
\newcommand{\attrlist}{\ensuremath{\mathsf{attrlist}}\xspace}
\newcommand{\fdlist}{\ensuremath{\mathsf{fdlist}}\xspace}
\newcommand{\fdclosure}{\ensuremath{\mathsf{fdclosure}}\xspace}

\newcommand{\X}{\mathbb{X}}
\newcommand{\scc}{\ensuremath{\mathsf{scc}}\xspace}

\AtBeginDocument{%
  }

\setcopyright{acmcopyright}
\copyrightyear{2023}
\acmYear{2023}
\acmDOI{XXXXXXX.XXXXXXX}

\begin{document}

\title[The Implication Problem for FDs, UMIs, and UMDEs.]{The Implication Problem for Functional Dependencies and Variants of Marginal Distribution Equivalences}

\author{Minna Hirvonen}
\email{minna.hirvonen@helsinki.fi}
\orcid{0000-0002-2701-9620}
\affiliation{%
  \institution{University of Helsinki}
  \streetaddress{PL 68 (Pietari Kalmin katu 5)}
  \city{Helsinki}
  \country{Finland}
  \postcode{00014}
}


\begin{abstract}
We study functional dependencies together with two different probabilistic dependency notions: unary marginal identity and unary marginal distribution equivalence. A unary marginal identity states that two variables $x$ and $y$ are identically distributed. A unary marginal distribution equivalence states that the multiset consisting of the marginal probabilities of all the values for variable $x$ is the same as the corresponding multiset for $y$. We present a sound and complete axiomatization for the class of these dependencies and show that  it has Armstrong relations. The axiomatization is infinite, but we show that there can be no finite axiomatization. The implication problem for the subclass that contains only functional dependencies and unary marginal identities can be simulated with functional dependencies and unary inclusion atoms, and therefore the problem is in polynomial-time. This complexity bound also holds in the case of the full class, which we show by constructing a polynomial-time algorithm. 
\end{abstract}





\maketitle

\section{Introduction}
Notions of dependence and independence are central to many areas of science. In database theory, the study of dependencies (or logical integrity constraints) is a central topic because it has applications to database design and many other data management tasks. For example, functional dependencies and inclusion dependencies are commonly used in practice as primary key and foreign key constraints. In this paper, we study functional dependencies together with unary marginal identity and unary marginal distribution equivalence. The latter two are probabilistic dependency notions that compare distributions of two variables. A unary marginal identity states that two variables $x$ and $y$ are identically distributed. A unary marginal distribution equivalence states that the
multiset consisting of the marginal probabilities of all the values for variable $x$ is
the same as the corresponding multiset for $y$. Marginal identity can actually be viewed as a probabilistic version of inclusion dependency; it is sometimes called ``probabilistic inclusion'' \cite{DurandHKMV18}.

We consider the so-called implication problem for the class of these dependencies. The \textit{implication problem} for a class of dependencies is the problem of deciding whether for a given finite set $\Sigma\cup\{\sigma\}$ of dependencies from the class, any database that satisfies every dependency from $\Sigma$ also satisfies $\sigma$. If the databases are required to be finite, the problem is called the \textit{finite implication problem}. Otherwise, we speak of the \textit{unrestricted implication problem}. We axiomatize the finite implication problem for functional dependence, unary marginal identity, and unary marginal distribution equivalence over uni-relational databases that are complemented with a probability distribution over the set of tuples appearing in the relation.

This paper is an extended version of a conference paper of the same title \cite{hirvonen22}. The main results are divided in sections as follows. In Section \ref{AxiomatizationSection}, we introduce the axiomatization and show that it is sound and complete. The axioms contain the so-called $k$-cycle rules for all odd natural numbers $k$, which make the axiomatization infinite. We also demonstrate that this is not due to an inappropriate choice of axioms, since the section includes a proof which shows that there can be no finite axiomatization. Section \ref{simulation_ch} contains results for the subclass without unary marginal distribution equivalences, i.e., the class that contains only functional dependencies and unary marginal identities. We show that for this subclass the implication problem can be simulated with functional dependencies and unary inclusion dependencies. Since the implication problem for functional dependencies and unary inclusion dependencies is known to be in polynomial time, we obtain from the simulation result a proof for the fact that the implication problem for the subclass is in polynomial time. In Section \ref{ComplexitySection}, we construct a polynomial time algorithm for the full class of functional dependencies and unary marginal identities and marginal distribution equivalences.

The implication problem that we axiomatize contains one qualitative class (functional dependence) and two probabilistic classes (marginal identity and marginal distribution equivalence) of dependencies, so we have to consider how these different kinds of dependencies interact. Some probabilistic dependencies have already been studied separately. The implication problem for probabilistic independence has been axiomatized over 30 years ago in \cite{geiger:1991}. More recently, the implication problem for marginal identity (over finite probability distributions) was axiomatized in \cite{hannula2021tractability}. The study of joint implication problems for different probabilistic and relational atoms have potential for practical applications because various probabilistic dependency notions appear in many situations. An example of an interesting class of probabilistic dependencies is probabilistic independencies together with marginal identities. In probability theory and statistics, random variables are often assumed to be independent and identically distributed (IID), a property which can be expressed with  probabilistic independencies and marginal identities. Another example comes from the foundations of quantum mechanics where functional dependence and probabilistic conditional independence can be used to express certain properties of hidden-variable models \cite{abramsky2021team},\cite{ALBERT2022103088}.

For practical purposes, it is important to consider implication problems also from a computational point of view: the usability of a class of dependencies, e.g. in database design, depends on the computational properties of its implication problem. We show that the implication problem for functional dependencies, unary marginal identities, and unary marginal distribution equivalences has a polynomial-time algorithm. 

A class of dependencies is said to have \textit{Armstrong relations} if for any finite set $\Sigma$ of dependencies from the class, there is a relation that satisfies a dependency $\sigma$ in the class if and only if $\sigma$ is true for every relation that satisfies $\Sigma$. An Armstrong relation can be used as a convenient representation of a dependency set. If a class of dependecies has Armstrong relations, then the implication problem for a fixed set $\Sigma$ and an arbitrary dependency $\sigma$ is reduced to checking whether the Armstrong relation for $\Sigma$ satisfies $\sigma$.  When Armstrong axiomatized functional dependence in \cite{armstrong74}, he also implicitly proved that the class has Armstrong relations \cite{Fagin82}. However, there are sets of functional dependecies for which the size of  a minimal Armstrong relation is exponential in the number of variables (attributes) of the set \cite{beeri84}. Armstrong relations can still be useful in practice \cite{LANGEVELDT2010352},\cite{mannila86}.

Sometimes integrity constraints, e.g. inclusion dependencies, are considered on multi-relational databases. In this case one looks for \textit{Armstrong databases} \cite{FAGIN198313} instead of single relations. In this terminology, an Armstrong relation is simply a uni-relational Armstrong database. Not all classes of dependencies enjoy Armstrong databases: functional dependence together with unary inclusion dependence (over multi-relational data\-bases where empty relations are allowed) does not have Armstrong databases \cite{FAGIN198313}. However, standard functional dependencies (i.e. functional dependencies with a non\-empty left-side) and inclusion dependencies do have Armstrong databases. It is known that probabilistic independence has Armstrong relations \cite{geiger:1991}. We show that the class of functional dependencies, unary marginal identities, and unary marginal distribution equivalences enjoys Armstrong relations.

Instead of working with notions and conventions from database theory, we have chosen to formulate the axiomatization in \textit{team semantics} which is a semantical framework for logics. This is because the dependency notions that we consider can be viewed as certain kinds of atomic formulas in logics developed for expressing probabilistic dependencies (see the logics studied e.g. in \cite{jelia19}), and we want to keep this connection\footnote{Our axiomatization is obviously only for the atomic level of these logics.} to these logics explicit. A ``team'' in team semantics is basically a uni-relational database, so the proofs that we present could easily also be stated in terms of databases.

Team semantics was originally introduced by Hodges \cite{hodges97}. The systematic development of the framework began with \textit{dependence logic}, a logic for functional dependence introduced by V{\"a}{\"a}n{\"a}nen \cite{vaananen07}, and the setting  has turned out to be useful for formulating logics for other dependency notions as well. These logics include, e.g., inclusion logic \cite{galliani12} and (conditional) independence logic \cite{gradel10}. In team semantics, logical formulas are evaluated over sets  of assignments (called \textit{teams}) instead of single assignments as, for example, in first-order logic. This allows us to define atomic formulas that express dependencies. For example, the dependency atom $\dep(\bar{x},\bar{y})$ expresses the functional dependency stating that the values of $\bar{x}$ determine the values of $\bar{y}$. As mentioned above, a team of assignments can be thought of as a relation (or a table) in a database: each assignment corresponds to a tuple in the relation (or a row in the table).

Since we want to study functional dependencies together with probabilistic dependency notions, we turn to \textit{probabilistic team semantics} which is a generalization of the relational team semantics. A \textit{probabilistic team} is a set of assignments with an additional function that maps each assignment to some numerical value, a \textit{weight}. The function is usually a probability distribution. As probabilistic team semantics is currently defined only for discrete distributions that have a finite number of possible values for variables, we consider the implication problem only for finite teams.

Although some probabilistic dependencies might seem similar to their qualitative variants, their implication problems are different: probabilistic dependencies refer to the weights of the rows rather than the rows themselves.  Probabilistic dependencies can be tricky, especially if one considers two or more different variants together. For example, consider marginal identity together with probabilistic independence. The chase algorithm that was used for proving the completeness of the axiomatization of marginal identity in \cite{hannula2021tractability} uses inclusion dependencies that contain index variables for counting multiplicities of certain tuples. There does not seem to be a simple way of extending this procedure to also include probabilistic independencies. This is because adding a new row affects the probability distribution and often breaks existing probabilistic independencies. On the other hand, the approach that was used for probabilistic independencies in \cite{geiger:1991} cannot easily be generalized to also cover inclusion dependencies either.

In our case, we can utilize the implication problem for functional dependencies and unary inclusion dependencies which has been axiomatized in \cite{CosmadakisKV90} for both finite and unrestricted implication. The axiomatization for finite implication is proved to be complete by constructing relations with the help of multigraphs that depict the dependecies between variables. This approach has also been applied to unary functional dependence, unary inclusion dependence, and independence \cite{hannula18}. Our approach is similar, but since we are working in the probabilistic setting, our construction must take into account the two kinds of variants of unary marginal distribution equivalences.

\section{Preliminaries}

Let $D$ be a finite set of variables and $A$ a finite set of values. We usually denote variables by $x,y,z$ and values by $a,b,c$. Tuples of variables and tuples of values are denoted by $\bar{x},\bar{y},\bar{z}$ and $\bar{a},\bar{b},\bar{c}$, respectively. The notation $|\bar{x}|$ means the length of the tuple $\bar{x}$, and $\var(\bar{x})$ means the set of variables that appear in the tuple $\bar{x}$.

An assignment of values from $A$ for the set $D$ is a function $s\colon D\to A$. A team $X$ of $A$ over $D$ is a finite set of assignments $s\colon D\to A$. When the variables of $D$ are ordered in some way, e.g. $D=\{x_1,\dots,x_n\}$, we identify the assignment $s$ with the tuple $s(x_1\dots x_n)=s(x_1)\dots s(x_n)\in A^{n}$, and also explicitly call $s$ a tuple. A team $X$ over $D=\{x_1,\dots,x_n\}$ can then be viewed as a table whose columns are the variables $x_1,\dots,x_n$, and rows are the tuples $s\in X$. For any tuple of variables $\bar{x}$ from $D$, we let
\[
X(\bar{x}):=\{s(\bar{x})\in A^{|\bar{x}|}\mid s\in X\}.
\]
A probabilistic team $\X$ is a function $\X\colon X\to (0,1]$ such that $\sum_{s\in X}(s)=1$. For a tuple of variables $\bar{x}$ from $D$ and a tuple of values $\bar{a}$ from $A$, we let
\[
|\X_{\bar{x}=\bar{a}}|:=\sum_{\substack{s(\bar{x})=\bar{a} \\ s\in X}}\X(s),
\]
i.e. $|\X_{\bar{x}=\bar{a}}|$ is the marginal probability of that the variables $\bar{x}$ have the values $\bar{a}$ in the probabilistic team $\X$.

Let $\bar{x},\bar{y}$ be tuples of variables from $D$. Then $\dep(\bar{x},\bar{y})$ is a \textit{(functional) dependency atom}. If the tuples $\bar{x},\bar{y}$ are of the same length, then $\bar{x}\subseteq\bar{y}$, $\bar{x}\approx\bar{y}$, and $\bar{x}\approx^*\bar{y}$ are \textit{inclusion (dependency)}, \textit{marginal identity}, and \textit{marginal distribution equivalence atoms}, respectively. We also use the abbreviations FD (functional dependency), IND (inclusion dependency), MI (marginal identity), and MDE (marginal distribution equivalence) for the atoms. If $|\bar{x}|=|\bar{y}|=1$, an atom is called \textit{unary} and abbreviated by UFD, UIND, UMI, or UMDE.

Before defining the semantics for the atoms, we need to introduce the notion of a \textit{multiset}. A multiset is a pair $(B,m)$ where $B$ is a set, and $m\colon B\to\mathbb{N}$ is a multiplicity function. The function $m$ determines for each element $b\in B$ how many multiplicities of $b$ the multiset $(B,m)$ contains. We often denote multisets using double wave brackets, e.g., $(B,m)=\{\{0,1,1\}\}$ when $B=\{0,1\}$,  and $m$ is such that $m(0)=1$ and $m(1)=2$.

Let $\sigma$ be one of the atoms described above. The notation $\X\models\sigma$ means that a probabilistic team $\X$ satisfies $\sigma$, which is defined as follows:
\begin{itemize}
\item[(i)] $\X\models\dep(\bar{x},\bar{y})$ iff for all $s,s'\in X$, if $s(\bar{x})=s'(\bar{x})$, then $s(\bar{y})=s'(\bar{y})$.
\item[(ii)] $\X\models\bar{x}\subseteq\bar{y}$ iff for all $s\in X$, there is $s'\in X$ such that $s(\bar{x})=s'(\bar{y})$.
\item[(iii)] $\X\models\bar{x}\approx\bar{y}$ iff $|\X_{\bar{x}=\bar{a}}|=|\X_{\bar{y}=\bar{a}}|$ for all $\bar{a}\in A^{|\bar{x}|}$.
\item[(iv)] $\X\models\bar{x}\approx^*\bar{y}$ iff $\{\{|\X_{\bar{x}=\bar{a}}|\mid \bar{a}\in X(\bar{x}) \}\}=\{\{|\X_{\bar{y}=\bar{a}}|\mid \bar{a}\in X(\bar{y}) \}\}$.
\end{itemize}
An atom $\dep(\bar{x},\bar{y})$ is called a functional dependency, because $\X\models\dep(\bar{x},\bar{y})$  iff there is a function $f\colon X(\bar{x})\to X(\bar{y})$ such that  $f(s(\bar{x}))=s(\bar{y})$ for all $s\in X$. An FD of the form $\dep(\lambda,x)$, where $\lambda$ is the empty tuple, is called a constant atom and denoted by $\dep(x)$. Intuitively, the constant atom $\dep(x)$ states that variable $x$ is constant in the team. An inclusion atom $\bar{x}\subseteq\bar{y}$ says that the set of values for $\bar{x}$ is included in the set of values for $\bar{y}$, i.e., $X(\bar{x})\subseteq X(\bar{y})$. The atom $\bar{x}\approx\bar{y}$ states that the tuples $\bar{x}$ and $\bar{y}$ give rise to identical distributions. The meaning of the atom $\bar{x}\approx^*\bar{y}$ is similar but allows the marginal probabilities to be attached to different tuples of values for $\bar{x}$ and $\bar{y}$.

\begin{example}
Consider the probabilistic team $\X$ depicted in Table \ref{simple_example}. The following are some examples of atoms that are (or are not) satisfied in $\X$. Functional dependencies: $\X\models\dep(x,w)$ holds because whenever two rows agree on a value of $x$, they also agree on $w$. Also $\X\models\dep(w)$ holds since the value of $w$ is 0 on every row. However, we have $\X\not\models\dep(x,y)$ because on the first two rows  the values of $x$ are both 0, but the value of $y$ is 0 on the first row and 1 on the second.

Inclusion dependencies: $\X\models x\subseteq y$ holds because all the values that appear in the $x$ column (i.e. 0 and 1) also appear in the $y$ column. Since the only values in the $y$ column are 0 and 1, we also have $\X\models y\subseteq x$.  However $\X\not\models x\subseteq w$ because the value 1 appears in $x$ column, but not in $w$. 

Marginal identities: $\X\models x\approx z$ holds because the marginal distributions of $x$ and $z$ are identical: probability of 0 is $2/3$ and the probability of 1 is $1/3$. Since the marginal distribution for $y$ is different: probability of 0 is $1/3$ and the probability of 1 is $2/3$, we have $\X\not\models x\approx y$.

Marginal distribution equivalences: $\X\models x\approx^* y$ and $\X\models x\approx^* z$ both hold because the multiset of marginal probabilities is the same $\{\{1/3,2/3\}\}$ for all three variables $x$, $y$, and $z$. However, $\X\not\models x\approx^* w$ because $\{\{1/3,2/3\}\}\neq\{\{1\}\}$. 

Note that the example demonstrates that a marginal identity implies the corresponding marginal distribution equivalence, but not the other way around.
\end{example}
\begin{table}[h]
	%\large
    \centering
    \caption{An example of a probabilistic team $\X$.}
    \begin{tabular}{ccccc}
        \toprule
        $x$ & $y$ & $z$ & $w$ & $\mathbb{X}$ \\
        \midrule
         $0$ & $0$ & $1$ & $0$ & 1/3 \\ 
		 $0$ & $1$ & $0$ & $0$ & 1/3 \\
		 $1$ & $1$ & $0$ & $0$ & 1/3 \\  
        \bottomrule
    \end{tabular}
     \label{simple_example}
\end{table}

Let $\Sigma\cup\{\sigma\}$ be a set of atoms. We write $\X\models\Sigma$ iff $\X\models\sigma'$ for all $\sigma'\in\Sigma$, and we write $\Sigma\models\sigma$ iff $\X\models\Sigma$ implies $\X\models\sigma$ for all $\X$. A decision problem of checking whether $\Sigma\models\sigma$ is called an \textit{implication problem}. We will always specify which classes of atoms are considered, e.g., if $\Sigma\cup\{\sigma\}$ is a set of FDs, UMIs, and UMDEs, the problem is called the implication problem for FDs+UMIs+UMDEs.

\section{Axiomatization for FDs+UMIs+UMDEs}\label{AxiomatizationSection}

In this section, we present an axiomatization for the implication problem of dependency, unary marginal identity, and unary marginal distribution equivalence atoms. The axiomatization is infinite and it will be shown in Section \ref{nonexistence} that there is no finite axiomatization.

\subsection{Axioms}\label{axioms}

The axioms for unary marginal identity and unary marginal distribution equivalence are the equivalence axioms of reflexivity, symmetry, and transitivity:
\begin{itemize}
\item[] UMI1: $x\approx x$
\item[] UMI2: If $x\approx y$, then $y\approx x$.
\item[] UMI3: If $x\approx y$ and $y\approx z$, then $x\approx z$.
\end{itemize}
\begin{itemize}
\item[] UMDE1: $x\approx^* x$
\item[] UMDE2: If $x\approx^* y$, then $y\approx^* x$.
\item[] UMDE3: If $x\approx^* y$ and $y\approx^* z$, then $x\approx^* z$.
\end{itemize}
For functional dependencies, we take the Armstrong axiomatization \cite{armstrong74} which consists of reflexivity, transitivity, and augmentation:
\begin{itemize}
\item[] FD1: $\dep(\bar{x},\bar{y})$ when $\var(\bar{y})\subseteq \var(\bar{x})$.
\item[] FD2: If $\dep(\bar{x},\bar{y})$ and $\dep(\bar{y},\bar{z})$, then $\dep(\bar{x},\bar{z})$.
\item[] FD3: If $\dep(\bar{x},\bar{y})$, then $\dep(\bar{x}\bar{z},\bar{y}\bar{z})$.
\end{itemize}
In Armstrong's original axiomatization, variables represent sets, whereas in this paper, we consider ordered sequences. Note that the difference is not important here because the condition $\var(\bar{y})\subseteq \var(\bar{x})$ in FD1 together with FD2 ensures that the order of variables in a tuple does not matter: any tuple $\bar{x}$ appearing in a functional dependency atom can be replaced by another tuple $\bar{x}'$ such that $\var(\bar{x})=\var(\bar{x}')$. 

Since marginal identity is a special case of marginal distribution equivalence, we have the following axiom:
\begin{itemize}
\item[] UMI \& UMDE: If  $x\approx y$, then $x\approx^* y$.
\end{itemize}
Let $k\in\{1,3,5,\dots\}$ and define function $S_k\colon\{0,1\dots,k\}\to\{0,1\dots,k\}$ as follows  
\[
S_k(i)=
\begin{cases}
i+1, \text{ when } 0\leq i<k\\
0, \text{ when } i=k.
\end{cases}
\]
For the unary functional dependencies and unary marginal distribution equivalencies, we have the \textit{k-cycle rule} for all $k\in\{1,3,5,\dots\}$:
\begin{align*}
&\text{If } \dep(x_i,x_{S_k(i)}) \text{ for all even } 0\leq i\leq k-1 \text{ and } x_j\approx^* x_{S_k(j)} \text{ for all odd } 1\leq j\leq k,\\
&\text{ then } \dep(x_{S_k(i)},x_i) \text{ and } x_i\approx^* x_{S_k(i)} \text{ for all even } 0\leq i\leq k-1
\end{align*}
The axiom is called a $k$-cycle rule because its antecedent forms a cycle in which FDs and UMDEs alternate. Note that the cycle rules are an infinite axiom schema, making our axiomatization infinite. 
The following inference rule is as a useful special case of the cycle rule
\begin{itemize}
\item[] UMDE \& FD: If $\dep(x,y)$ and $\dep(y,x)$, then $x\approx^* y$.
\end{itemize}
This can be seen by noticing that by the cycle rule
\begin{align*}
\dep(x,y) \text{ and } y\approx^* y \text{ and } \dep(y,x) \text{ and } x\approx^* x
\end{align*}
implies $x\approx^* y$. 

Axioms that allow us to infer new atoms from old ones are sometimes called ``inference rules'' in contrast to axioms that state that a certain atom always holds, e.g., UMI2 is an inference rule and UMI1 is an axiom. In this paper, this distinction is not made, instead the terms ``inference rule'' and ``axiom'' are used interchangeably.

Let $\Sigma\cup\{\sigma\}$ be a set of atoms and $\mathcal{A}$ an axiomatization, i.e., a set of axioms. Let $D$ be the set of variables that appear in $\Sigma$. We write $\Sigma\vdash_{\mathcal{A}}\sigma$ iff $\sigma$ can be derived from $\Sigma$ by using the axiomatization $\mathcal{A}$. We denote by $\cl_{\mathcal{A}}(\Sigma)$ the set of atoms obtained by closing $\Sigma$ under the axioms $\mathcal{A}$, i.e., for all $\sigma$ with variables from $D$, $\sigma\in\cl_{\mathcal{A}}(\Sigma)$ iff $\Sigma\vdash_{\mathcal{A}}\sigma$. We say that an axiomatization $\mathcal{A}$ is \textit{sound} if for any set $\Sigma\cup\{\sigma\}$ of dependencies, $\Sigma\vdash_{\mathcal{A}}\sigma$ implies $\Sigma\models\sigma$. Correspondingly, an axiomatization is \textit{complete} if for any set $\Sigma\cup\{\sigma\}$ of dependencies, $\Sigma\models\sigma$ implies $\Sigma\vdash_{\mathcal{A}}\sigma$. If the axiomatization $\mathcal{A}$ is clear from the context, we write $\Sigma\vdash\sigma$ and $\cl(\Sigma)$ instead of $\Sigma\vdash_{\mathcal{A}}\sigma$ and $\cl_{\mathcal{A}}(\Sigma)$.

The axiomatization defined above is sound. We only show that the cycle rules are sound; since it is straightforward to check the soundness of other the axioms, we leave out their proofs.

We notice that  for all $x,y\in D$, $\X\models\dep(x,y)$ implies $|X(x)|\geq |X(y)|$ and $\X\models x\approx^* y$ implies $|X(x)|= |X(y)|$. Suppose now that the antecedent of the $k$-cycle rule holds for $\X$. It then follows that 
\[
|X(x_0)|\geq|X(x_1)|=|X(x_2)|\geq\dots=|X(x_{k-1})|\geq|X(x_{k})|=|X(x_{0})|, 
\]
which implies that $|X(x_i)|=|X(x_j)|$ for all $i,j\in \{0,\dots k\}$. If $i,j$ are additionally such that $\X\models\dep(x_i,x_j)$, then there is a surjective function $f\colon X(x_i)\to X(x_j)$ for which $f(s(x_i))=s(x_j)$ for all $s\in X$. Since $X(x_i)$ and $X(x_j)$ are both finite and have the same number of elements, the function $f$ is also one-to-one. Therefore the inverse of $f$ is also a function, and we have $\X\models\dep(x_j,x_i)$, as wanted. Since $f$ is bijective, there is a one-to-one correspondence between $X(x_i)$ and $X(x_j)$. Thus $|\X_{x_i=a}|=|\X_{x_j=f(a)}|$ for all $a\in X(x_i)$, and we have $\{\{|\X_{x_i=a}|\mid a\in X(x_i) \}\}=\{\{|\X_{x_j=a}|\mid a\in X(x_j) \}\}$, which implies that $\X\models x_i\approx^* x_j$.

Note that $k$-cycle rules are not in general sound for equivalence relations. Consider the equivalence relation $\approx_{\max}$ such that $\X\models x\approx_{\max} y$ if and only if $\max\{|\X_{x=a}|\mid a\in X(x) \}=\max\{|\X_{y=a}|\mid a\in X(y) \}$. It is easy to see that $\approx_{\max}$ is indeed an equivalence relation. However, the probabilistic team $\X$ in Table \ref{counter_example} demonstrates that $\{\dep(x_0,x_1), x_1\approx_{\max} x_2, \dep(x_2,x_3), x_3\approx_{\max} x_0\}\not\models\dep(x_1,x_0)$, i.e., the $3$-cycle rule for UFDs and $\approx_{\max}$-atoms is not sound. This is because $\max\{|\X_{x_i=a}|\mid a\in X(x_i) \}=1/2$ for all $0\leq i\leq 3$, and therefore both of the atoms $x_1\approx_{\max} x_2$ and $x_3\approx_{\max} x_0$ are satisfied. Since it is easy to check that $\dep(x_0,x_1)$ and $\dep(x_2,x_3)$ are satisfied, but $\dep(x_1,x_0)$ is not, we obtain  $\X\models\{\dep(x_0,x_1), x_1\approx_{\max} x_2, \dep(x_2,x_3), x_3\approx_{\max} x_0\}$ and $\X\not\models\dep(x_1,x_0)$.

\begin{table}[h]
	%\large
    \centering
    \caption{The probabilistic team $\X$.}
    \begin{tabular}{ccccc}
        \toprule
        $x_0$ & $x_1$ & $x_2$ & $x_3$ & $\mathbb{X}$ \\
        \midrule
         $0$ & $0$ & $0$ & $0$ & 1/4 \\ 
		 $1$ & $0$ & $0$ & $0$ & 1/4 \\
		 $2$ & $1$ & $1$ & $1$ & 1/4 \\
		 $2$ & $1$ & $2$ & $1$ & 1/4 \\  
        \bottomrule
    \end{tabular}
     \label{counter_example}
\end{table}

\subsection{Completeness and Armstrong relations}
\label{completeness&armstrong_rel}

In this section, we show that our axiomatization is complete and the class of FDs, UMIs and UMDEs has Armstrong relations. Recall that for any finite set $\Sigma$ of dependencies from the class, an Armstrong relation  is a probabilistic team $\X$ such that for every $\sigma$ from the class, $\X\models\sigma$ if and only if $\Sigma\models\sigma$.

We prove the completeness by showing that for any set $\Sigma$ of FDs, UMIs, and UMDEs, there is a probabilistic team $\X$ such that $\X\models\sigma$ iff $\sigma\in\cl(\Sigma)$. Note that proving this implies completeness: if $\sigma\not\in\cl(\Sigma)$ (i.e. $\Sigma\not\vdash\sigma$), then $\X\not\models\sigma$. Since $\X\models\Sigma$, we have $\Sigma\not\models\sigma$. By doing the proof this way, we obtain Armstrong relations because by completeness and soundness, $\sigma\in\cl(\Sigma)$ iff $\Sigma\models\sigma$, from which it follows that the constructed team $\X$ is an Armstrong relation for $\Sigma$. Note that a probabilistic team is actually not a relation but a probability distribution. Therefore, to be exact, instead of Armstrong relations, we should speak of \textit{Armstrong models}, which is a more general notion introduced in \cite{Fagin82}. In our setting, the Armstrong models we construct are uniform distributions over a relation, so each model is determined by a relation, and it suffices to speak of Armstrong relations. 

We consider the following variant of multigraphs with colored edges:
\begin{definition}
Let $n\in\mathbb{N}$. A mixed multigraph with colored edges is a tuple $G=(V,E_1,\dots,E_n)$, where $V$ is the set of vertices and each $E_i$, $1\leq i\leq n$, is either 
\begin{itemize}
\item[(i)] a set of ordered pairs $(x,y)\in V\times V$ (directed edge) or 
\item[(ii)] a set of subsets $\{x,y\}\subseteq V$ (undirected edge).
\end{itemize}
The members of each $E_i$ are called edges of the color $i$, and each $(V,E_i)$ is called the $i$-colored subgraph of $G$.
\end{definition}

We now define a multigraph in which different-colored edges correspond to different types of dependencies between variables. By using the properties of the graph, we can construct a suitable team, which can then be made into a probabilistic team by taking the uniform distribution over the assignments.

\begin{definition}\label{graph_def}
Let $\Sigma$ be a set of FDs, UMIs, and UMDEs. We define a multigraph $G(\Sigma)$ as follows: 
\begin{itemize}
\item[(i)] the set of vertices consists of the variables appearing in $\Sigma$,
\item[(ii)] for each marginal identity $x\approx y\in\Sigma$, there is an undirected black edge between $x$ and $y$,
\item[(iii)] for each marginal distribution equivalence $x\approx^* y\in\Sigma$, there is an undirected blue edge between $x$ and $y$, and
\item[(iv)] for each functional dependency $\dep(x,y)\in\Sigma$, there is a directed red edge from $x$ to $y$.
\end{itemize}
\end{definition}
If there are red directed edges both from $x$ to $y$ and from $y$ to $x$, they can be thought of as a single red undirected edge between $x$ and $y$.

\begin{example}\label{example}
Let $\Delta=\{\dep(x_0,x_1),x_1\approx x_2,\dep(x_2,x_3),x_3\approx^*x_0,\dep(x_4),\dep(x_5)\}$, and $\Sigma=\cl(\Delta)$. The graph $G(\Sigma)$ constructed as in Definition \ref{graph_def} is depicted Figure \ref{example_fig}.
\end{example}

%\begin{figure}
%  \centering
%  \includegraphics[scale=0.175]{pic1}
%  \caption{The graph $G(\Sigma)$ of Example \ref{example}. For the sake of clarity, we have removed from $G(\Sigma)$ all self-loops and some edges that are implied by transitivity.}
%  \label{example_fig}
%  \end{figure}
  
\begin{figure}
%\large
  \centering
  \begin{tikzpicture}[node distance={15mm},thick, main/.style = {draw, circle}] 
\node[main] (0) {$x_0$}; 
\node[main] (1) [right of=0] {$x_1$};
\node[main] (2) [below of=1] {$x_2$};
\node[main] (3) [below of=0] {$x_3$};
\node[main] (4) [right of=1] {$x_4$};
\node[main] (5) [below of=4] {$x_5$};
%red
\draw[->] [color=red] (1) -- (4);
\draw[->] [color=red] (2) -- (5);
\draw [color=red] (0) to [out=45,in=135,looseness=1] (1);
\draw [color=red] (3) to [out=-45,in=-135,looseness=1] (2);
\draw [color=red] (4) to [out=-45,in=45,looseness=1] (5);
%blue
\draw [color=blue] (0) -- (1);
\draw [color=blue] (3) -- (2);
\draw [color=blue] (0) -- (3);
\draw [color=blue] (1) -- (2);
\draw [color=blue] (4) -- (5);
%black
\draw (1) to [out=-135,in=135,looseness=1] (2);

\end{tikzpicture} 
  \caption{The graph $G(\Sigma)$ of Example \ref{example}. For the sake of clarity, we have removed from $G(\Sigma)$ all self-loops and some edges that are implied by transitivity.}
  \label{example_fig}
  \end{figure}
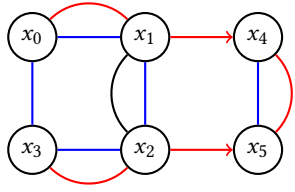

\begin{definition}[Strongly connected component]
A multigraph $G=(V,E_1,\dots,E_n)$ is \textit{strongly connected} if for all $x,y\in V$ $y$ is reachable from $x$, i.e., there is a sequence of vertices $v_1,v_2,\dots,v_m\in V$ such that $v_1=x$, $v_m=y$, and for all $1\leq j\leq m$, $(v_j,v_{j+1})\in E_i$ or $\{v_j,v_{j+1}\}\in E_i$ for some $1\leq i\leq n$. A set of vertices $S\subseteq V$ is said to be a \textit{strongly connected component} of $G$ if its induced subgraph is a maximal strongly connected subgraph of $G$.
% with respect to the subgraph relation.
%The set of vertices of a subgraph $G'$ of $G$ is said to be a \textit{strongly connected component} of $G$ if $G'$ is a maximal strongly connected subgraph with respect to the subgraph relation.
\end{definition}

\begin{definition}[Clique]
Let $G$ a multigraph $(V,E_1,\dots,E_n)$ and $i$ be such that $E_i$ is a set of undirected edges. An \textit{i-colored clique} in $G$ is a set of vertices $C\subseteq V$ such that for all $x,y\in C$, $\{x,y\}\in E_i$.
\end{definition}

In order to use the above graph construction to show the existence of Armstrong relations, we need to consider certain properties of the graph. 
\begin{lemma}\label{graph_prop_lemma}
Let $\Sigma$ be a set of FDs, UMIs and UMDEs that is closed under the inference rules, i.e. $\cl(\Sigma)=\Sigma$. Then the graph $G(\Sigma)$ has the following properties:
\begin{itemize}
\item[(i)] Every vertex has a black, blue, and red self-loop.
\item[(ii)] The black, blue, and red subgraphs of $G(\Sigma)$ are all transitively closed. 
\item[(iii)] The black subgraph of $G(\Sigma)$ is a subgraph of the blue subgraph of $G(\Sigma)$.
\item[(iv)] The subgraphs induced by the strongly connected components of $G(\Sigma)$ are undirected. Each such component contains a black, blue, and red undirected subgraph. In the subgraph of each color, the vertices of the component can be partitioned into a collection of disjoint cliques.
All the vertices of the component belong to a single blue clique.
\item[(v)] If $\dep(\bar{x},y)\in\Sigma$ and the vertices $\bar{x}$ have a common ancestor $z$ in the red subgraph of $G(\Sigma)$, then there is a red edge from $z$ to $y$.
\end{itemize}
\end{lemma} 
\begin{proof}
\begin{itemize}
\item[(i)] By the reflexivity rules UMI1, UMDE1, and FD1, atoms $x\approx x$, $x\approx^* x$, and $\dep(x,x)$ are in $\Sigma$ for every vertex $x$. 
\item[(ii)] The transitivity rules UMI3, UMDE3, and FD2 ensure the transitivity of the black, blue, and red subgraphs, respectively. 
\item[(iii)] By the UMI \& UMDE rule, we have that if there is a black edge between $x$ and $y$, then there is also a blue edge between $x$ and $y$.

\item[(iv)]In a strongly connected component every black or blue edge is undirected by the definition of $G(\Sigma)$. Consider then a red edge from $x$ to $y$. Since we are in a strongly connected component, there is also a path from $y$ to $x$. By (iii), we may assume that this path only consist of blue and red edges. By adding the red edge from $x$ to $y$ to the beginning of the path, we obtain a cycle from $x$ to $x$. (Note that in this cycle, some vertices might be visited more than once.) By using the part (i), we can add blue and red self-loops, if necessary, to make sure that the cycle is constructed from alternating blue and red edges as in the cycle rule. From the corresponding $k$-cycle rule, it then follows that there is also a red edge from $y$ to $x$. 

The existence of a partition follows from the fact that in each strongly connected component the black/blue/red edges define an equivalence relation for the vertices. (The reflexivity and transitivity follow from parts (i) and (ii), and the fact that each strongly connected component is undirected implies symmetry.) Lastly, pick any vertices $x$ and $y$ in a  strongly connected component. We claim that there is a blue edge between $x$ and $y$. Since we are in a strongly connected component, there is a path from $x$ to $y$. The path consists of undirected black, blue and red edges. By (iii), each black edge of the path can be replaced with a blue one. Similarly, by the rule UMDE \& FD, each undirected red edge can be replaced with a blue one. Thus, there is a blue path from $x$ to $y$, and, by transitivity, also a blue edge.

\item[(v)] Let $\bar{x}=x_0\dots x_n$. Each vertex $x_j$, $0\leq j\leq n$, has a common ancestor $z$ in the red subgraph of $G(\Sigma)$ and the red subgraph is transitive by part (ii), so $z$ is also a direct ancestor of each $x_j$, i.e. $\dep(z,x_j)\in\Sigma$. Suppose that we have both $\dep(z, x_0\dots x_{k})$ and $\dep(z,x_{k+1})$ in $\Sigma$ for some $0\leq k\leq n-1$. By FD1 and FD3, we have $\dep(z,zz)$, $\dep(zz,x_{k+1}z)$, and $\dep(x_{k+1}z,zx_{k+1})$ in $\Sigma$. Then by using FD2 twice, we obtain $\dep(z,zx_{k+1})\in\Sigma$. Next, by FD3, we also have $\dep(zx_{k+1},x_0\dots x_{k}x_{k+1})\in\Sigma$, and then by using FD2, we obtain $\dep(z, x_0\dots x_{k+1})\in\Sigma$. Since we have $\dep(z, x_j)\in\Sigma$ for all $0\leq j\leq n$, this shows that $\dep(z,\bar{x})\in\Sigma$. By using FD2 to $\dep(z,\bar{x})\in\Sigma$ and $\dep(\bar{x},y)\in\Sigma$, we obtain $\dep(z,y)\in\Sigma$, i.e., there is a red edge from $x$ to $y$.
\end{itemize}
\end{proof}

\begin{example}\label{example2}
Consider the graph $G(\Sigma)$ of Example $\ref{example}$ depicted in Figure \ref{example_fig}. The strongly connected components of $G(\Sigma)$ are $\{x_0,x_1,x_2,x_3\}$ and $\{x_4,x_5\}$, the maximal red cliques are $\{x_0,x_1\}$, $\{x_2,x_3\}$, and $\{x_4,x_5\}$, the maximal blue cliques are $\{x_0,x_1,x_2,x_3\}$ and $\{x_4,x_5\}$ and the maximal black cliques are $\{x_0\}$, $\{x_1,x_2\}$, $\{x_3\}$, $\{x_4\}$, and $\{x_5\}$.
\end{example}

Let $S$ and $S'$ be two different strongly connected components of a graph. We say that $S'$ is a \textit{descendant component} of $S$ if for every $v\in S$ and $v'\in S'$, the vertex $v'$ is a descendant of $v$. Since this means that for every $v\in S$ and $v'\in S'$, the vertex $v$ is an ancestor of $v'$, we also say that $S$ is an \textit{ancestor component} of $S'$.

We next give a unique number to each strongly connected component of $G(\Sigma)$. The numbers are assigned such that the number of a descendant component is always greater than the number of its ancestor component. We call these numbers scc-numbers and denote by $\scc(x)$ the scc-number of the strongly connected component that the vertex $x$ belongs to. Note that if $G(\Sigma)$ has two strongly connected components without an edge between them, there exist several possible scc-numberings. In this case, we just choose one of the possible numberings as our fixed scc-numbering.

\begin{definition}\label{construction2}
Let $\Sigma$ and $G(\Sigma)$ be as in Lemma \ref{graph_prop_lemma} and assign an scc-numbering to  $G(\Sigma)$. Let $D$ be the set of the variables appearing in $\Sigma$. We define a probabilistic team $\X$ as the uniform distribution over $X$, where the team $X$ over the variables $D$ is constructed as follows:
\begin{itemize}
\item[(i)] Add a tuple of all zeroes, i.e., an assignment $s$ such that $s(x)=0$ for all $x\in D$.
\item[(ii)] Process each strongly connected component in turn, starting with the one with the smallest scc-number and then proceeding in the ascending order of the numbers. For each strongly connected component, handle each of its maximal red cliques in turn.
\begin{itemize}
\item[(a)] For each maximal red clique $k$, add a tuple with zeroes in the columns corresponding to the variables in $k$ and to the variables that are in any red clique that is a red descendant of $k$. Leave all the other positions in the tuple empty for now.
\item[(b)] Choose a variable in $k$ and count the number of zeroes in a column corresponding to the chosen variable. It suffices to consider only one variable, because the construction ensures that all the columns corresponding to variables in $k$ have the same number of zeroes. Denote this number by $\text{count}(k)$.
\end{itemize}
After adding one tuple for each maximal red clique, check that the $\text{count}(k)$ is equal for every clique $k$ in the current component and strictly greater than the count of each clique in the previous component. If it is not, repeat some of the tuples added to make it so. This can be done, because a red clique $k$ can be a red descendant of another red clique $j$ only if $j$ is in a strongly connected component with a strictly smaller scc-number than the one of the component that $k$ is in, and thus $j$'s component is already processed. Note that the counts of the cliques in a component do not change after the component has been processed.
\item[(iii)] The last component is a single red clique consisting of those variables $x$ for which $\dep(x)\in\Sigma$, if any. Each variable in this clique functionally depends on all the other variables in the graph, so we do not leave any empty positions in its column. Therefore the columns corresponding to these variables contain only zeroes. If there are no variables $x$ for which $\dep(x)\in\Sigma$, we finish processing the last component by adding one tuple with all positions empty.
\item[(iv)] After all strongly connected components have been processed, we start filling the empty positions. Process again each strongly connected component in turn, starting with the one with the smallest scc-number. For each strongly connected component, count the number of maximal black cliques. If there are $n$ such cliques, number them from 0 to $n-1$. Then handle each maximal black clique $k$, $0\leq k\leq n-1$ in turn.
\begin{itemize}
\item[(a)] For each column in clique $k$, count the number of empty positions. If the column has $d>0$ empty positions, fill them with numbers $1,\dots d-1,d+k$ without repetitions. (Note that each column in $k$ has the same number of empty positions.) 
\end{itemize}
\item[(v)] If there are variables $x$ for which $\dep(x)\in\Sigma$, they are all in the last component. The corresponding columns contain only zeroes and have no empty positions. As before, count the number of maximal black cliques. If there are $n$ such cliques, number them from 0 to $n-1$. Then handle each maximal black clique $k$, $0\leq k\leq n-1$ in turn.
\begin{itemize}
\item[(a)] For each column in clique $k$, change all the zeroes into $k$'s. 
\end{itemize}

After handling all the components (including the one consisting of constant columns, if there are any), there are no empty positions anymore, and the construction is finished.
\end{itemize}
\end{definition}

The zero tuple is added in step (i) to make sure that any variable $x$ determines functionally only those variables that appear in the same maximal red clique as $x$ or in some its red descendant cliques. This can be achieved by adding in step (ii), for each strongly connected component, new tuples where zeroes appear in positions according to the maximal red cliques. This is described in more detail in the proof of Lemma \ref{basic_lem2} (i) \& (ii).

The purpose of counting in (iib) is to ensure that the columns for variables in the same strongly connected component have the same number of zeroes. Since every pair of variables in the same strongly connected component has an UMDE between them, the same number of zeroes ensures that when we later fill in the empty positions without repeating values, these UMDEs will hold in the team. Note that by this counting, we also ensure that columns that correspond to variables in different strongly connected components have different numbers of zeroes. This ensures that only the inferable UMDEs are satisfied.

When we fill in the empty positions in step (iv), the chosen values are determined by the maximal black cliques, so only the inferable UMIs hold. Note that if there are constant variables, their values will also be renamed in step (v) for the same reason. 

\begin{example}\label{example3}
Let $\Sigma$ be as in Example \ref{example2} and the probabilistic team $\X$ as in Definition \ref{construction2}. The team $\X$ is depicted in Table \ref{example_table}.
\end{example}

\begin{table}[h]
	%\large
    \centering
    \caption{The probabilistic team $\X$ of Example \ref{example3} constructed as in Definition \ref{construction2}.}
    \begin{tabular}{ccccccc}
        \toprule
        $x_0$ & $x_1$ & $x_2$ & $x_3$ &$x_4$ & 				$x_5$ & $\mathbb{X}$ \\
        \midrule
         $0$ & $0$ & $0$ & $0$ & $0$ & $1$ & 1/4 \\ 
		$0$ & $0$ & $1$ & $1$ & $0$ & $1$ & 1/4 \\
		 $1$ & $1$ & $0$ & $0$ & $0$ & $1$ & 1/4 \\
 		$2$ & $3$ & $3$ & $4$ & $0$ & $1$ & 1/4 \\ 
        \bottomrule
    \end{tabular}
     \label{example_table}
\end{table}

The next lemma shows that the team we constructed above already has many of the wanted properties:
\begin{lemma}\label{basic_lem2}
Let $\Sigma$, $G(\Sigma)$, and $\X$ be as in Definition \ref{construction2}. Then the following statements hold:
\begin{itemize}
\item[(i)] For any nonunary functional dependency $\sigma$, if $\sigma\in\Sigma$, then $\X\models\sigma$.
\item[(ii)] For any unary functional dependency or constancy atom $\sigma$, $\X\models\sigma$ if and only if $\sigma\in\Sigma$.
\item[(iii)] For any unary marginal distribution equivalence atom $\sigma$, $\X\models\sigma$ if and only if $\sigma\in\Sigma$.
\item[(iv)] For any unary marginal identity atom $\sigma$, $\X\models\sigma$ if and only if $\sigma\in\Sigma$.
\end{itemize}
\end{lemma}
\begin{proof}
\begin{enumerate}
\item[(i)] Suppose that $\sigma=\dep(\bar{x},\bar{y})\in\Sigma$. We may assume that $\bar{x}$,$\bar{y}$ do not contain variables $z$ for which $\dep(z)\in\Sigma$. This is because for any such variable $z$, $\X\models\dep(\bar{x},\bar{y})$ implies $\X\models\dep(\bar{x}z,\bar{y})$ and $\X\models\dep(\bar{x},\bar{y}z)$. By FD1, we have $\dep(\bar{x}z,\bar{x})$, which together with $\dep(\bar{x},\bar{y})$ implies $\dep(\bar{x}z,\bar{y})$ by FD2. From $\dep(z)$, we obtain $\dep(\bar{y},\bar{y}z)$ by FD3. Then by applying FD2 to $\dep(\bar{x},\bar{y})$ and $\dep(\bar{y},\bar{y}z)$, we obtain $\dep(\bar{x},\bar{y}z)$.

For a contradiction suppose that $\X\not\models\sigma$. Then there are tuples $s,s'$ that violate $\sigma$. By our assumption, the only repeated number in each column of variables $\bar{x}$,$\bar{y}$ is 0. Thus $s(\bar{x})=s'(\bar{x})=\bar{0}$, and either $s(\bar{y})\neq\bar{0}$ or $s'(\bar{y})\neq\bar{0}$. Assume that $s(\bar{y})\neq\bar{0}$. Then there is $y_0\in\var(\bar{y})$ such that $s(y_0)\neq 0$. By our construction, tuple $s$ corresponds to a red clique. This means that in $s$, each variable that is 0, is functionally determined by every variable of the corresponding red clique. Hence, there is a variable $z$ such that $s(z)=0$ and $z$ is a common ancestor of vertices $\bar{x}$ in the red subgraph. As $\dep(\bar{x},\bar{y})\in\Sigma$ and, by FD1, $\dep(\bar{y},y_0)\in\Sigma$, we obtain $\dep(\bar{x},y_0)\in\Sigma$ by using FD2. Then by Lemma \ref{graph_prop_lemma} (v), we have $\dep(z,y_0)\in\Sigma$, and thus, by the construction, $s(y_0)=0$, which is a contradiction.

\item[(ii)] Suppose that $\sigma=\dep(x)$. If $\sigma\in\Sigma$, then $x$ is determined by all the other variables in $\Sigma$. This means that $x$ is in the last component (the one with the highest scc-number), and thus the only value appearing in column $x$ is $k$, where $k$ is the number assigned to the maximal black clique of $x$. Hence, we have $\X\models\sigma$. If $\sigma\not\in\Sigma$, then either $x$ is not in the last component or there are no variables $y$ for which $\dep(y)\in\Sigma$. In either case, we have at some point added a tuple in which the position of $x$ is empty, and therefore there are more than one value appearing in column $x$, i.e., $\X\not\models\sigma$.

Suppose then that $\sigma=\dep(x,y)$. We may assume that neither $\dep(x)$ nor $\dep(y)$ are in $\Sigma$. If $\dep(y)\in\Sigma$, then $y$ is constant, and we have $\sigma\in\Sigma$ and $\X\models\sigma$. If $\dep(x)\in\Sigma$ and $\dep(y)\not\in\Sigma$, then $\sigma\not\in\Sigma$ and $\X\not\models\sigma$. 
If $\sigma\in\Sigma$, we are done since this is a special case of (i). If $\sigma\not\in\Sigma$, then the first tuple and the tuple that was added for the red clique of $x$ violate $\sigma$.

\item[(iii)] Suppose that $\sigma=x\approx^* y$. Since $\X$ is obtained by taking the uniform distribution over $X$, $\X\models x\approx^* y$ holds if and only if $|X(x)|=|X(y)|$. The latter happens if and only if $x$ and $y$ are in the same strongly connected component. Since any strongly connected component is itself the maximal blue clique, this is equivalent to $x$ and $y$ being in the same maximal blue clique in a strongly connected component, i.e., $x\approx^* y\in\Sigma$.

\item[(iv)] Suppose that $\sigma=x\approx y$. Since $\X$ is obtained by taking the uniform distribution over $X$, $\X\models x\approx y$ holds if and only if $X(x)=X(y)$. The latter happens if and only if $x$ and $y$ are in the same maximal black clique in a strongly connected component. This happens if and only if $x\approx y\in\Sigma$.
\end{enumerate}

\end{proof}

The construction of Definition \ref{construction2} does not yet give us Armstrong relations. This is because there might be nonunary functional dependencies $\sigma$ such that $\X\models\sigma$ even though $\sigma\not\in\Sigma$, as demonstrated in the following example. 

\begin{example}\label{example4}
Let $\Delta=\{x_0\approx^*x_1,x_1\approx^*x_2\}$ and $\Sigma=\cl(\Delta)$. The graph $G(\Sigma)$ and probabilistic team $\X$ for $G(\Sigma)$ constructed as in Definition \ref{construction2} are depicted in Figure \ref{ex4pic} and Table \ref{ex4table}, respectively. Note that, e.g., $\X\models\dep(x_0x_1,x_2)$, but $\dep(x_0x_1,x_2)\not\in\Sigma$.
\end{example}

%\begin{figure}
%  \centering
%   \includegraphics[scale=0.21]{pic3}
%  \caption{The graph $G(\Sigma)$ of Example \ref{example4}.}
%  \label{ex4pic}
%\end{figure}

\begin{figure}
  
  \centering
  \begin{tikzpicture}[node distance={15mm},thick, main/.style = {draw, circle}] 
\node[main] (0) {$x_0$}; 
\node[main] (1) [right of=0] {$x_1$};
\node[main] (2) [below of=1] {$x_2$};
\draw [color=blue] (0) -- (1);
\draw [color=blue] (0) -- (2);
\draw [color=blue] (1) -- (2);
\end{tikzpicture}
  \caption{The graph $G(\Sigma)$ of Example \ref{example4}.}
  \label{ex4pic}
\end{figure}
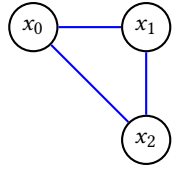

\begin{table}[h]
	%\large
    \centering
    \caption{The probabilistic team $\X$ of Example \ref{example4} constructed as in Definition \ref{construction2}.}
\centering
  \begin{tabular}{ccccccc} 
 \toprule
 $x_0$ & $x_1$ & $x_2$ & $\mathbb{X}$ \\
 \midrule
 $0$ & $0$ & $0$ & 1/4 \\ 
 $0$ & $1$ & $1$ & 1/4 \\
 $1$ & $0$ & $4$ & 1/4 \\
 $2$ & $3$ & $0$ & 1/4 \\   
 \bottomrule
\end{tabular}
  \label{ex4table}
\end{table}

In the lemma below, we show how to construct a probabilistic team for which $\X\models\sigma$ implies $\sigma\in\Sigma$ also for nonunary functional dependencies $\sigma$.

\begin{lemma}\label{FD_lemma2}
Let $\Sigma$ and $G(\Sigma)$ be as in Lemma \ref{graph_prop_lemma}. Then there exists a probabilistic team $\X$ such that $\X\models\sigma$ if and only if $\sigma\in\Sigma$.
\end{lemma}
\begin{proof}
Let $D$ be the set of the variables appearing in $\Sigma$. We again begin by constructing a relational team $X$ over $D$, but this time we modify the construction from Definition \ref{construction2} by adding a new step (0), which will be processed first, before continuing with step (i). From step (i), we will proceed as in Definition \ref{construction2}. Let $\Sigma'$ be a set of all nonunary functional dependencies $\dep(\bar{x},\bar{y})\not\in\Sigma$ where $\bar{x},\bar{y}\in D$.
The step (0) is then defined as follows:
\begin{itemize}
\item[(0)] For each $C\subseteq D$ with $|C|\geq 2$, check if there is $\dep(\bar{x},\bar{y})\in\Sigma'$ such that $\var(\bar{x})=C$. If yes, add a tuple with zeroes in exactly those positions $z$ that are functionally determined by $\bar{x}$, i.e., those $z$ for which $\dep(\bar{x},z)\in\Sigma$. Leave all the other positions in the tuple empty for now.
\end{itemize}
As before, we define the probabilistic team $\X$ as the uniform distribution over $X$. We first show that $\X$ violates all atoms from $\Sigma'$. Let $\sigma=\dep(\bar{x},\bar{y})\in\Sigma'$. We may assume that $\bar{x}$,$\bar{y}$ do not contain variables $z$ for which $\dep(z)\in\Sigma$ because $\X\not\models\dep(\bar{x},\bar{y})$ implies $\X\not\models\dep(\bar{x}z,\bar{y})$ and $\X\not\models\dep(\bar{x},\bar{y}z)$ for any such variable $z$. Since $\sigma\not\in\Sigma$, there is $y_0\in\var(\bar{y})$ such that $y_0$ is not determined by $\bar{x}$, i.e., $\dep(\bar{x},y_0)\not\in\Sigma$. Thus, the tuple that was added for the set $C=\var(\bar{x})$ in step (0), and the tuple added in step (i), agree on variables $\bar{x}$ (they are all zeroes), but in the first tuple, variable $y_0$ is nonzero and in the other tuple, $y_0$ is zero. Hence, we have $\X\not\models\sigma$.

We still need to show that $\X\models\Sigma$, and that $\sigma\not\in\Sigma$ implies $\X\not\models\sigma$, when $\sigma$ is a UMI,  UMDE, UFD or a constancy atom. Note that since the construction proceeds from step (i) as in Definition \ref{construction2}, it ensures that the counts of the empty positions are the same for all the columns that correspond to the variables in the same strongly connected component. Therefore, for any unary marginal identity or distribution equivalence atom $\sigma$, $\X\models\sigma$ if and only if $\sigma\in\Sigma$ as shown in Lemma \ref{basic_lem2}, parts (iii) \& (iv). For a unary functional dependency, or a constancy atom $\sigma\not\in\Sigma$, we obtain, by Lemma \ref{basic_lem2}, part (ii), that already the part of the team $\X$ that was constructed as in Definition \ref{construction2} contains tuples that violate $\sigma$.

 If $\dep(z)\in\Sigma$, then clearly $\dep(\bar{x},z)\in\Sigma$ for any $\bar{x}$, and thus adding the tuples in step (0) does not end up violating $\dep(z)$. Suppose then that $\dep(\bar{z},\bar{z}')\in\Sigma$, where $\bar{z}$ and $\bar{z}'$ can also be single variables. For a contradiction, suppose there are two tuples that together violate $\dep(\bar{z},\bar{z}')$. We may again assume that $\bar{z}$,$\bar{z}'$ do not contain variables $u$ for which $\dep(u)\in\Sigma$. By our assumption, the only repeated numbers in columns $\bar{z}$ are zeroes, so it has to be the case that the variables $\bar{z}$ are all zeroes in the two tuples that violate $\dep(\bar{z},\bar{z}')$. We now show that in any such tuple, the variables $\bar{z}'$ are also all zeroes, contradicting the assumption that there are two tuples that agree on $\bar{z}$ but not on $\bar{z}'$.  
 
Suppose that the tuple was added on step (0). By the construction, we then have some variables $\bar{x}$,$\bar{y}$ such that $|\bar{x}|\geq 2$, $\dep(\bar{x},\bar{y})\in\Sigma'$, and $\dep(\bar{x},\bar{z})\in\Sigma$. Since $\dep(\bar{z},\bar{z}')\in\Sigma$, by transitivity, we have $\dep(\bar{x},\bar{z}')\in\Sigma$, and thus the variables $\bar{z}'$ are also all zeroes in the tuple, as wanted. Suppose then that the tuple was not added in step (0), i.e., it was added for some maximal red clique in a strongly connected component. Suppose then that there is $z_0\in\var(\bar{z}')$ which is not zero in the tuple. This leads to a contradiction as shown in Lemma \ref{basic_lem2} part (i). Hence, the tuple has to have all zeroes for variables $\bar{z}'$.
\end{proof}

Note that the size of the constructed Armstrong relations may be large with respect to the number of variables $|D|:=n$ because in the construction of Lemma \ref{FD_lemma2}, we may need to add up to $2^{n}-(n+2)$ new tuples to ensure that the nonunary FDs that are not implied by $\Sigma$ are not satisfied.

\begin{example}
Let $\Sigma$  and $G(\Sigma)$ be as in Example \ref{example4}. Table \ref{ex4table2} depicts the probabilistic team $\X$ constructed as in Lemma \ref{FD_lemma2}. Note that now for all functional dependencies $\sigma$, $\X\models\sigma$ if and only if $\sigma\in\Sigma$.
\end{example}
\begin{table}
\caption{The probabilistic team $\X$ of Example \ref{example4} constructed as in Lemma \ref{FD_lemma2}.}
\begin{tabular}{ ccccccc } 
 \toprule
 $x_0$ & $x_1$ & $x_2$ & $\mathbb{X}$ \\
 \midrule 
 $0$ & $0$ & $1$ & 1/7 \\
 $0$ & $1$ & $0$ & 1/7 \\
 $1$ & $0$ & $0$ & 1/7 \\
 $0$ & $0$ & $0$ & 1/7 \\ 
 $0$ & $2$ & $2$ & 1/7 \\
 $2$ & $0$ & $5$ & 1/7 \\
 $3$ & $4$ & $0$ & 1/7 \\   
 \bottomrule
\end{tabular}
\label{ex4table2}
\end{table}

\subsection{Non-existence of a complete $k$-ary axiomatization}\label{nonexistence}

An axiom whose antecedent contains $k$ atoms, i.e., an axiom of the form 
\[
(\sigma_1\text{ and }\dots\text{ and }\sigma_k)\implies(\delta_1\text{ and }\dots\text{ and }\delta_l)\]
is called $k$-ary. We say that an axiomatization is $k$-ary if each of its axioms is at most $k$-kary. In this section, we show that for any $k$, there is no complete $k$-ary axiomatization for the implication problem for FDs+UMIs+UMDEs, which means that any complete axiomatization has to be infinite. The approach is similar to the one used in \cite{HANNULA16} in the case of finite implication problem for (relational) independence atoms and keys.

Fix $k\in\{3,5,\dots\}$ and consider the set
\[
\Sigma_k=\{\dep(x_0,x_1),x_1\approx^* x_2,\dep(x_2,x_3),\dots,x_{k-2}\approx^*x_{k-1},\dep(x_{k-1},x_k),x_k\approx^* x_0\},
\]
which consists of the atoms from the antecedent of $k$-cycle rule. Denote $\bar{x}=x_0\dots x_k$, and define the following set:
\begin{align*}
\Sigma=\Sigma_k&\cup\{x_0\approx x_0,x_1\approx x_1,\dots,x_k\approx x_k\}\\
&\cup\{x_0\approx^* x_0,x_1\approx^* x_1,\dots,x_k\approx^* x_k\}\\
&\cup\{x_0\approx^* x_k,x_{k-1}\approx^* x_{k-2},\dots,x_{4}\approx^* x_{3},x_{2}\approx^* x_{1}\}\\
&\cup\Sigma',
\end{align*}
where $\Sigma'$ is defined inductively as follows: 
\begin{itemize}
\item[(i)] If $\var(\bar{z})\subseteq\var(\bar{y})\subseteq\var(\bar{x})$, then $\dep(\bar{y},\bar{z})\in\Sigma'$.
\item[(ii)] If $\dep(\bar{y},\bar{z})\in\Sigma'$ and $i\in\{0,2,\dots,k-1\}$, then $\dep(\bar{y}x_i,\bar{z}x_{i+1})\in\Sigma'$.
\item[(iii)] If $\dep(\bar{y},\bar{z})\in\Sigma'$, $\var(\bar{y})=\var(\bar{y}')$, and $\var(\bar{z})=\var(\bar{z}')$, then $\dep(\bar{y}',\bar{z}')\in\Sigma'$.
\end{itemize}
Let $\cl_k(\Sigma_k)$ be the closure of $\Sigma_k$ under $k$-ary implication, i.e., $\sigma\in\cl_k(\Sigma_k)$ iff $\Sigma_k\vdash_{\mathcal{A}}\sigma$, where $\mathcal{A}$ is the set of all sound axioms that are at most $k$-ary.

We claim that $\Sigma=\cl_k(\Sigma_k)$. By the $k$-cycle rule (which is a $k+1$-ary axiom) and the completeness, we have $\Sigma_k\models\dep(x_1,x_0)$. By the definition of $\Sigma'$, we have $\dep(x_1,x_0)\not\in\Sigma'$. Since $\dep(x_1,x_0)\not\in\Sigma_k$, we have $\dep(x_1,x_0)\not\in\Sigma$. Hence, from $\Sigma=\cl_k(\Sigma_k)$, it follows that there is no complete $k$-ary axiomatization for the FD+UMI+UMDE implication.

For $\Sigma=\cl_k(\Sigma_k)$, we prove the following two lemmas:
\begin{lemma}\label{k-axiom_lemma1}
$\Sigma\subseteq\cl_k(\Sigma_k)$.
\end{lemma}
\begin{proof}
Clearly $\Sigma_k\subseteq\cl_k(\Sigma_k)$. The second, the third, and the fourth set of the union are subsets of $\cl_k(\Sigma_k)$ because they are obtained from $\Sigma_k$ by using the rules UMI1, UMDE1, and UMDE2, respectively. Suppose then that $\sigma\in\Sigma'$. The cases (i) and (iii) of the inductive definition of $\Sigma'$ are clear by FD1 and FD2. Thus, we may assume that $\sigma=\dep(\bar{y}x_{i},\bar{z}x_{i+1})$ for some $i\in\{0,2,\dots,k-1\}$ and $\bar{y},\bar{z}$ such that $\dep(\bar{y},\bar{z})\in\cl_k(\Sigma_k)$. As $\dep(x_{i},x_{i+1})\in\Sigma_k\subseteq\cl_k(\Sigma_k)$, it suffices to show that for all $\bar{v},\bar{v}',\bar{w},\bar{w}'$, we can infer from the atoms $\dep(\bar{v},\bar{w})$ and $\dep(\bar{v}',\bar{w}')$, the atom $\dep(\bar{v}\bar{v}',\bar{w}\bar{w}')$ by using at most $k$-ary axioms. This can be done as follows: by applying FD3 to $\dep(\bar{v},\bar{w})$ and $\dep(\bar{v}',\bar{w}')$, we obtain $\dep(\bar{v}\bar{v}',\bar{w}\bar{v}')$ and $\dep(\bar{v}'\bar{w},\bar{w}'\bar{w})$. Then by applying FD1 to $\dep(\bar{v}'\bar{w},\bar{w}'\bar{w})$, we obtain $\dep(\bar{w}\bar{v}',\bar{w}\bar{w}')$. Then by applying FD2 to $\dep(\bar{v}\bar{v}',\bar{w}\bar{v}')$ and $\dep(\bar{w}\bar{v}',\bar{w}\bar{w}')$, we obtain $\dep(\bar{v}\bar{v}',\bar{w}\bar{w}')$, as wanted. 
%Since the proof is similar to the case (v) of Lemma \ref{graph_prop_lemma}, we omit it.
\end{proof}
\begin{lemma}\label{k-axiom_lemma2}
$\cl_k(\Sigma_k)\subseteq \Sigma$.
\end{lemma} 
\begin{proof} We prove the claim by induction on the length of the deduction of the atoms $\sigma\in\cl_k(\Sigma_k)$.

The case $n=0$ is clear, since $\Sigma_k\subseteq\Sigma$.

For the case $n=m+1$, let $\sigma\in\cl_k(\Sigma_k)$ be such that the length of its deduction from $\Sigma_k$ is $n$. Let $\delta_1,\dots,\delta_k\in\cl_k(\Sigma_k)$ be the atoms that appear in the last step of deduction. (Note that it might be that $\delta_i=\delta_j$ for some $i\neq j$.) By the induction hypothesis $\delta_1,\dots,\delta_k\in\Sigma$. Denote $\Delta:=\{\delta_1,\dots,\delta_k\}$. Suppose then for a contradiction that $\sigma\not\in\Sigma$. We will show that then $\Delta\not\models\sigma$. Let $\delta\in\Sigma_k\backslash\Delta$. Such $\delta$ exists because $\Sigma_k$ has $k+1$ different atoms, whereas $\Delta$ only has at most $k$ atoms. Note that by symmetry of the cycle, it suffices to consider the cases  $\delta=\dep(x_0,x_1)$ and $\delta=x_k\approx^*x_0$. 

Consider first the case: $\delta=\dep(x_0,x_1)$. Since $\Delta\subseteq\Sigma\backslash\{\delta\}$, we have $\Sigma\backslash\{\delta\}\models\Delta$. Define the set $A:=\{\dep(x_{0}x_{i},x_{1}x_{i+1})\mid i\in\{2,4,\dots,k-1\}\}$. It can be easily checked that $\Sigma_k\backslash\{\delta\}\cup A\models\Sigma\backslash\{\delta\}$. This means that $\Sigma_k\backslash\{\delta\}\cup A\models\Delta$, so it suffices to show that $\Sigma_k\backslash\{\delta\}\cup A\not\models\sigma$.

Hence, we construct a probabilistic team $\X$ that satisfies $\Sigma_k\backslash\{\delta\}\cup A$ but not $\sigma$. We will see that we can obtain $\X$ by constructing a probabilistic team for the set $\cl(\Sigma_k\backslash\{\delta\})$ as described in Definition \ref{construction2} and Lemma \ref{basic_lem2}. By the construction, we have $\X\models\Sigma_k\backslash\{\delta\}$, so it still remains to check that $\X\models A$. Suppose that $\X\not\models\dep(x_{0}x_{i},x_{1}x_{i+1})$. Since the only repeating value in the columns is zero, the atom $\dep(x_{0}x_{i},x_{1}x_{i+1})$ has to be violated by two rows in which the values of $x_0$ and $x_i$ are both zero. The vertex $x_0$ forms its own maximal red clique, which has no red descendants. (See Figure \ref{nonexfig}.) Hence, by the construction, there can be only one row (the first one) which has zeroes in both columns $x_0$ and $x_i$. Therefore, $\X\models A$.

Now, we show that $\X\not\models\sigma$. First, suppose that $\sigma$ is a constancy atom, UFD, UMI, or UMDE. Then, by the construction, $\X\models\sigma$ if and only if the graph $G(\cl(\Sigma_k\backslash\{\delta\}))$ has the corresponding edge. (Note that since $\Sigma_k\backslash\{\delta\}$ does not imply any constancy atoms, $\X$ does not satisfy any constancy atoms either.) From Figure \ref{nonexfig}, it is easy to see that every edge in $G(\cl(\Sigma_k\backslash\{\delta\}))$ has the corresponding atom in $\Sigma$. This means that $\X\models\sigma$ implies $\sigma\in\Sigma$. Since $\sigma\not\in\Sigma$, we have $\X\not\models\sigma$, and hence $\Sigma_k\backslash\{\delta\}\cup A\not\models\sigma$.

Suppose then that $\sigma$ is a nonunary FD. We modify the previous construction in a way that is similar to Lemma \ref{FD_lemma2}. The difference is that, instead of adding a row for each nonunary FD that is not implied by the atoms, we only need to add one for $\sigma$. In other words, if $\sigma=\dep(\bar{y},\bar{z})$, then we add a tuple with zeroes in exactly those positions $u$ that are functionally determined by $\bar{y}$, i.e., those for which $\cl(\Sigma_k\backslash\{\delta\})\cup A\models\dep(\bar{y},u)$. The construction then follows as in Lemma \ref{FD_lemma2}. Denote the constructed probabilistic team by $\X'$. Now by Lemma \ref{FD_lemma2}, $\X'\models\Sigma_k\backslash\{\delta\}$ and $\X'\not\models\sigma$. 

It remains to show that $\X'\models A$. Suppose that $\X'\not\models\dep(x_{0}x_{i},x_{1}x_{i+1})$. As before, the only repeating value in the columns is zero, so the atom $\dep(x_{0}x_{i},x_{1}x_{i+1})$ has to be violated by two rows in which the values of $x_0$ and $x_i$ are both zero. These two rows have to be the row that contains only zeroes and the new row that was added. If the new added row has zeroes in columns $x_0$ and $x_i$, both variables have to be functionally determined by $\bar{y}$. Since we check whether $u$ are functionally determined by $\bar{y}$ with respect to the whole set $\cl(\Sigma_k\backslash\{\delta\})\cup A$ and $\dep(x_{0}x_{i},x_{1}x_{i+1})\in A$, variables $x_1$ and $x_{i+1}$ are also functionally determined by $\bar{y}$, and both have value zero. Hence, the two rows cannot violate the atom, and it must be that $\X'\models A$. Therefore, $\Sigma_k\backslash\{\delta\}\cup A\not\models\sigma$.

Consider then the case: $\delta=x_k\approx^*x_0$. Now, $\Delta\subseteq\Sigma\backslash\{\delta, x_0\approx^*x_k\}$, so we have $\Sigma\backslash\{\delta, x_0\approx^*x_k\}\models\Delta$. Since in this case $\Sigma_k\backslash\{\delta\}\models\Sigma\backslash\{\delta, x_0\approx^*x_k\}$, it suffices to show that $\Sigma_k\backslash\{\delta\}\not\models\sigma$. The proof is omitted as it is very similar  to (albeit simpler than) the one in the previous case.
\end{proof}

%\begin{figure}
% \centering
%	\includegraphics[scale=0.31]{pic2}
%  \caption{The graph $G(\cl(\Sigma_k\backslash\{\delta\}))$ in the two cases (a) $\delta=\dep(x_0,x_1)$ and (b) $\delta=x_k\approx^*x_0$. For the sake of clarity, all the self-loops have been removed.}
%  \label{nonexfig}
%\end{figure}

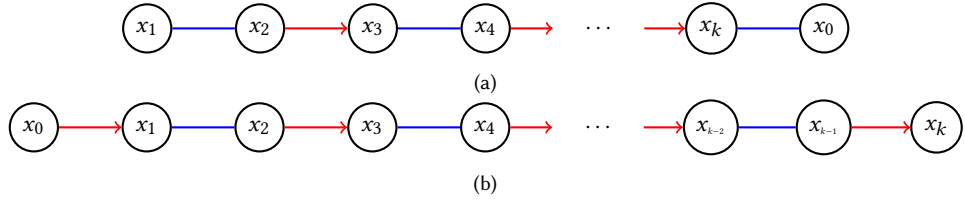
\begin{figure}
 \centering
 \begin{subfigure}[a]{1\textwidth}
 \centering
  \begin{tikzpicture}[node distance={15mm},thick, main/.style = {draw, circle}] 
\node[main] (1) {$x_1$}; 
\node[main] (2) [right of=1] {$x_2$};
\node[main] (3) [right of=2] {$x_3$};
\node[main] (4) [right of=3] {$x_4$};
\node (5) [right of=4] {$\quad\dots\quad$};
\node[main] (6) [right of=5] {$x_k$};
\node[main] (0) [right of=6] {$x_0$};
%red
\draw[->] [color=red] (2) -- (3);
\draw[->] [color=red] (4) -- (5);
\draw[->] [color=red] (5) -- (6);
%blue
\draw [color=blue] (1) -- (2);
\draw [color=blue] (3) -- (4);
\draw [color=blue] (6) -- (0);
%black
%\draw (1) to [out=-135,in=135,looseness=1] (2);
\end{tikzpicture} 
\caption{}
\end{subfigure}
\begin{subfigure}[b]{1\textwidth}
\centering
  \begin{tikzpicture}[node distance={15mm},thick, main/.style = {draw, circle}] 
  \node[main] (0) {$x_0$};
\node[main] (1) [right of=0]{$x_1$}; 
\node[main] (2) [right of=1] {$x_2$};
\node[main] (3) [right of=2] {$x_3$};
\node[main] (4) [right of=3] {$x_4$};
\node (5) [right of=4] {$\quad\dots\quad$};
\node[main] (6) [right of=5] {$x_{\scaleto{k-2}{2.5pt}}$};
\node[main] (7) [right of=6] {$x_{\scaleto{k-1}{2.5pt}}$};
\node[main] (8) [right of=7] {$x_{k}$};
%red
\draw[->] [color=red] (0) -- (1);
\draw[->] [color=red] (2) -- (3);
\draw[->] [color=red] (4) -- (5);
\draw[->] [color=red] (5) -- (6);
\draw[->] [color=red] (7) -- (8);
%blue
\draw [color=blue] (1) -- (2);
\draw [color=blue] (3) -- (4);
\draw [color=blue] (6) -- (7);
%black
%\draw (1) to [out=-135,in=135,looseness=1] (2);
\end{tikzpicture} 
\caption{}
\end{subfigure}
  \caption{The graph $G(\cl(\Sigma_k\backslash\{\delta\}))$ in the two cases (a) $\delta=\dep(x_0,x_1)$ and (b) $\delta=x_k\approx^*x_0$. For the sake of clarity, all the self-loops have been removed.}
  \label{nonexfig}
\end{figure}
From Lemmas \ref{k-axiom_lemma1} and \ref{k-axiom_lemma2}, it follows that $\Sigma=\cl_k(\Sigma_k)$, and therefore no sound $k$-ary axiomatization can be complete. Since any finite axiomatization would be $k$-ary for some natural number $k$, we obtain the following theorem:
\begin{theorem}
There is no finite (complete) axiomatization for functional dependencies, unary marginal identities, and unary marginal distribution equivalences.
\end{theorem}

\section{Simulating FD+UMI implication with FD+UIND implication}\label{simulation_ch}

In this section, we show that the implication problem for FDs+UMIs can be simulated with the implication problem for FDs+UINDs by simulating each marginal identity $x\approx y$ with two inclusion atoms $x\subseteq y$ and $y\subseteq x$. Since the implication problem for FDs+UINDs has a polynomial-time algorithm \cite{CosmadakisKV90}, this means that the same algorithm can also be used for FD+UMI implication, and therefore the latter problem is also in polynomial-time.

\subsection{Axioms for FD+UMI and FD+UIND}

Axioms UMI1--UMI3, FD1--FD3, and the following variant of the $k$-cycle rules for all $k\in\{1,3,5,\dots\}$: 
\begin{align*}
&\text{If } \dep(x_i,x_{S_k(i)}) \text{ for all even } 0\leq i\leq k-1 \text{ and }x_j\approx x_{S_k(j)} \text{ for all odd } 1\leq j\leq k,\\
&\text{ then } \dep(x_{S_k(i)},x_i) \text{ for all even } 0\leq i\leq k-1,
\end{align*}
%\begin{align*}
%&\text{If } \dep(x_0,x_1) \text{ and } x_1\approx x_2 \text{ and }\dots \text{ and } \dep(x_{k-1},x_k) \text{ and } x_k\approx x_0,\\
%&\text{ then } \dep(x_1,x_0) \text{ and } \dots \text{ and } \dep(x_k,x_{k-1})
%\end{align*}
form a sound and complete axiomatization for the implication problem for FD+UMI. Soundness of this variant of the cycle rule immediately follows from the fact that $x\approx y$ implies $x\approx^*  y$. For the completeness of the axiomatization, we can construct the Armstrong relations as in Section \ref{completeness&armstrong_rel}.

The finite implication problem for FD+UIND as been axiomatized in \cite{CosmadakisKV90}. The axioms are listed below. For unary inclusion dependecies, we have reflexivity and transitivity.
\begin{itemize}
\item[] UIND1: $x\subseteq x$
\item[] UIND2: If $x\subseteq y$ and $y\subseteq z$, then $x\subseteq z$.
\end{itemize}
For functional dependencies, we take the Armstrong axiomatization (axioms FD1--FD3) from Section \ref{axioms}.

For the unary functional dependencies and unary inclusion dependencies, we have the \textit{k-cycle rules} for all $k\in\{1,3,5,\dots\}$:
\begin{align*}
&\text{If } \dep(x_i,x_{S_k(i)}) \text{ for all even } 0\leq i\leq k-1 \text{ and } x_{S_k(j)}\subseteq x_{j} \text{ for all odd } 1\leq j\leq k\\
&\text{ then } \dep(x_{S_k(i)},x_{i}) \text{ for all even } 0\leq i\leq k-1 \text{ and } x_{j}\subseteq x_{S_k(j)} \text{ for all odd } 1\leq i\leq k.
\end{align*}
%\begin{align*}
%&\text{If } \dep(x_0,x_1) \text{ and } x_2\subseteq x_1 \text{ and }\dots \text{ and } \dep(x_{k-1},x_k) \text{ and } x_0\subseteq x_k,\\
%&\text{ then } \dep(x_1,x_0) \text{ and } x_1\subseteq x_2 \text{ and }\dots \text{ and } \dep(x_k,x_{k-1}) \text{ and } x_k\subseteq x_0
%\end{align*}
 
The soundness of $k$-cycle rules for unary functional dependencies and unary inclusion dependencies was shown in \cite{CosmadakisKV90} but we include a brief proof for the sake of self-containment.

Suppose now that the antecedent of the $k$-cycle rule holds for $\X$. It then follows that 
\[
|X(x_0)|\geq|X(x_1)|\geq|X(x_2)|\geq\dots\geq|X(x_{k-1})|\geq|X(x_{k})|\geq|X(x_{0})|, 
\]
which implies that $|X(x_i)|=|X(x_j)|$ for all $i,j\in \{0,\dots k\}$. If $i,j$ are such that $\X\models\dep(x_i,x_j)$, then $\X\models\dep(x_j,x_i)$ follows from the same argument as in the case of $k$-cycle rules for unary functional dependencies and unary marginal distribution equivalence in Section \ref{axioms}. If $i,j$ are such that $\X\models x_i\subseteq x_j$, then $X(x_i)\subseteq X(x_j)$. When combined with the fact that $|X(x_i)|=| X(x_j)|$, we obtain that $X(x_j)\subseteq X(x_i)$, i.e., $\X\models x_j\subseteq x_i$.

%If $i,j$ are additionally such that $\X\models\dep(x_i,x_j)$, then there is a surjective function $f\colon X(x_i)\to X(x_j)$ for which $f(s(x_i))=s(x_j)$ for all $s\in X$. Since $X(x_i)$ and $X(x_j)$ are both finite and have the same number of elements, the function $f$ is also one-to-one. Therefore the inverse of $f$ is also a function, and we have $\X\models\dep(x_j,x_i)$, as wanted. Since $f$ is bijective, there is a one-to-one correspondence between $X(x_i)$ and $X(x_j)$. Thus $|\X_{x_i=a}|=|\X_{x_j=f(a)}|$ for all $a\in X(x_i)$, and we have $\{\{|\X_{x_i=a}|\mid a\in X(x_i) \}\}=\{\{|\X_{x_j=a}|\mid a\in X(x_j) \}\}$, which implies that $\X\models x_i\approx^* x_j$.

\subsection{Simulation}

The following theorem shows that FD+UMI implication can be simulated with FD+UIND implication:

\begin{theorem}\label{simulation}
Let $\Sigma_{\text{FD}}$ and $\Sigma_{\text{UMI}}$ be sets of functional dependencies and unary marginal identities, respectively. Denote $\Sigma=\Sigma_{\text{FD}}\cup\Sigma_{\text{UMI}}$, and define $\Sigma^*=\Sigma_{\text{FD}}\cup\{x\subseteq y\mid x\approx y\in\Sigma_{\text{UMI}} \text{ or } y\approx x\in\Sigma_{\text{UMI}}\}$. Then for all functional dependencies $\sigma$ and unary marginal identities $x\approx y$
\begin{itemize}
\item[(i)] $\Sigma\models\sigma$ if and only if $\Sigma^*\models\sigma$,
\item[(ii)] $\Sigma\models x\approx y$ if and only if $\Sigma^*\models x\subseteq y$ and $\Sigma^*\models y\subseteq x$.
\end{itemize}
\end{theorem}
\begin{proof}
Note that since both axiomatizations are sound and complete, we can substitute ``$\models$'' with ``$\vdash$'' in the above theorem. 

In the proof of Theorem \ref{simulation}, we will use the following Lemma, which we prove first:
\begin{lemma}\label{uind_lemma}
Let $\Sigma$ and $\Sigma^*$ be as in Theorem \ref{simulation}. Suppose that $\Sigma^*\vdash x\subseteq y$. Then $\Sigma^*\vdash y\subseteq x$ and the length of the deduction of $y\subseteq x$ is at most that of  $x\subseteq y$.
\end{lemma}
\begin{proof}
We prove the claim by induction on the length of the deduction.
\begin{itemize}
\item $n=0$. This means that $x\subseteq y\in\Sigma^*$, so the case is clear by the definition of $\Sigma^*$.
\item $n=m+1$. The case of the reflexivity axiom ``$x\subseteq x$'' is clear.

Suppose that the last axiom used in the deduction of  $x\subseteq y$ is transitivity. Then there is some variable $z$ such that  $\Sigma^*\vdash x\subseteq z$ and $\Sigma^*\vdash z\subseteq y$. From the induction hypothesis, we obtain $\Sigma^*\vdash y\subseteq z$ and $\Sigma^*\vdash z\subseteq x$. Then by the transitivity of the inclusion atoms, $\Sigma^*\vdash y\subseteq x$. Note also that by the induction hypothesis, the  length of the deduction of $y\subseteq x$ is at most that of  $x\subseteq y$.

Suppose then that the last axiom used in the deduction of  $x\subseteq y$ is the cycle rule. Then the claim follows trivially since $y\subseteq x$ appears already in the assumptions of the rule.
\end{itemize}

\end{proof}

Now, we are ready to prove Theorem \ref{simulation}. The proof is again by induction on the length of the deduction. For notational convenience, we use the logical conjunction ``$\wedge$'' in the usual meaning, i.e., the notation $\Sigma\vdash\sigma_0\wedge\sigma_1$ means that $\Sigma\vdash\sigma_0$ and $\Sigma\vdash\sigma_1$. 

We begin by showing that the ``from left to right'' sides of the two items in the theorem hold.
\begin{itemize}
\item $n=0$. This means that $\sigma\in\Sigma$ or $x\approx y\in\Sigma$. This case is clear by the definitions of $\Sigma$ and $\Sigma^*$.
\item $n=m+1$. The case of the reflexivity axiom ``$x\approx x$'' is clear since ``$x\subseteq x$'' is an axiom in the other axiomatization.

Suppose that the last axiom used in the deduction of  $x\approx y$ is symmetry. Then $\Sigma\vdash y\approx x$, and by the induction hypothesis $\Sigma^*\vdash y\subseteq x$ and $\Sigma^*\vdash x\subseteq y$.

Suppose that the last axiom used in the deduction of  $x\approx y$ is transitivity. Then there is some variable $z$ such that  $\Sigma\vdash x\approx z$ and  $\Sigma\vdash z\approx y$. From the induction hypothesis, we obtain $\Sigma^*\vdash x\subseteq z$, $\Sigma^*\vdash z\subseteq y$, $\Sigma^*\vdash y\subseteq z$, and $\Sigma^*\vdash z\subseteq x$. Then by the transitivity of the inclusion atoms, $\Sigma^*\vdash x\subseteq y$ and $\Sigma^*\vdash y\subseteq x$.

The cases of the FD axioms are clear since all of these axioms are in both of the axiomatizations.

Suppose that the last axiom used in the deduction of  $\dep(x,y)$ is the cycle rule. Then $\Sigma\vdash \dep(x_0,x_1)\wedge x_1\approx x_2\wedge \dots \wedge \dep(x_{k-1},x_k)\wedge x_k\approx x_0$, where $\dep(x,y)\in\{\dep(x_1,x_0),\dots,\dep(x_k,x_{k-1})\}$. By the induction hypothesis, we have $\Sigma^*\vdash \dep(x_0,x_1)\wedge x_2\subseteq x_1\wedge \dots \wedge \dep(x_{k-1},x_k)\wedge x_0\subseteq x_k$. From the cycle rule, we obtain $\Sigma^*\vdash \dep(x_1,x_0)\wedge x_1\subseteq x_2\wedge \dots \wedge \dep(x_k,x_{k-1})\wedge x_k\subseteq x_0$.

\end{itemize}

Next, we prove the ``from right to left'' sides of the items.
\begin{itemize}
\item $n=0$. This means that $\sigma\in\Sigma^*$ or both $x\subseteq y, y\subseteq x\in\Sigma^*$. This case is clear by the definitions of $\Sigma$ and $\Sigma^*$.
\item $n=m+1$. The case of the reflexivity axiom ``$x\subseteq x$'' is clear since ``$x\approx x$''  is an axiom in the other axiomatization.

Suppose that the last axiom used in the deduction of  $x\subseteq y$ is transitivity. Then there is some variable $z$ such that  $\Sigma^*\vdash x\subseteq z$ and $\Sigma^*\vdash z\subseteq y$. From Lemma \ref{uind_lemma}, we obtain $\Sigma^*\vdash z\subseteq x$ and $\Sigma^*\vdash y\subseteq z$, and that the lengths of the deductions of $z\subseteq x$ and $y\subseteq z$ are at most that of $x\subseteq z$ and $z\subseteq y$, respectively. Now $\Sigma^*\vdash x\subseteq z\wedge z\subseteq x$ and $\Sigma^*\vdash z\subseteq y\wedge y\subseteq z$. Thus by the induction hypothesis, we have $\Sigma\vdash x\approx z\wedge z\approx y$. Then by transitivity, we have $\Sigma\vdash x\approx y$.

The cases of the FD axioms are clear since all of these axioms are in both of the axiomatizations.

Suppose that the last axiom used in the deduction of  $\dep(x,y)$ is the cycle rule. Then $\Sigma^*\vdash \dep(x_0,x_1)\wedge x_2\subseteq x_1\wedge \dots \wedge \dep(x_{k-1},x_k)\wedge x_0\subseteq x_k$, where $\dep(x,y)\in\{\dep(x_1,x_0),\dots,\dep(x_k,x_{k-1})\}$. By Lemma \ref{uind_lemma}, we have $\Sigma^*\vdash x_1\subseteq x_2\wedge \dots \wedge x_k\subseteq x_0$, so by the induction hypothesis, we have $\Sigma\vdash \dep(x_0,x_1)\wedge x_1\approx x_2\wedge \dots \wedge \dep(x_{k-1},x_k)\wedge x_k\approx x_0$. Hence, by the cycle rule, we obtain $\Sigma\vdash \dep(x_1,x_0)\wedge \dots \wedge \dep(x_k,x_{k-1})$.

Suppose that the last axiom used in the deduction of  $x\subseteq y$ is the cycle rule. Then $\Sigma^*\vdash \dep(x_0,x_1)\wedge x_2\subseteq x_1\wedge \dots \wedge \dep(x_{k-1},x_k)\wedge x_0\subseteq x_k$, where $x\subseteq y\in\{x_1\subseteq x_2,\dots,x_k\subseteq x_0\}$. By Lemma \ref{uind_lemma}, we have $\Sigma^*\vdash x_1\subseteq x_2\wedge \dots \wedge x_k\subseteq x_0$, where the length of the deduction for each inclusion dependence is at most that of the corresponding reversed inclusion, i.e., at most of the length $m$. Since $y\subseteq x$ already appeared in the antecedent of the used rule, this means that there is already a shorter deduction for $x\subseteq y$ (the one obtained from Lemma \ref{uind_lemma}), and thus by the induction hypothesis $\Sigma\vdash x\approx y$.
\end{itemize}
\end{proof}

From Theorem \ref{simulation} and Lemma \ref{uind_lemma}, we obtain the following corollary:
\begin{corollary}
The implication problem for functional dependencies and unary marginal identities is in polynomial time.
\end{corollary}
\begin{proof}
For a set $\Sigma\cup\{\sigma\}$ of functional dependencies and unary marginal identities, let $\Sigma^*$ be as in Theorem \ref{simulation}, and let 
\[
\sigma^*:=
\begin{cases}
\sigma \quad &\text{ if } \sigma \text{ is an FD},\\
x\subseteq y \quad &\text{ if } \sigma=x\approx y.
\end{cases}
\]
By Theorem \ref{simulation} and Lemma \ref{uind_lemma}, to determine whether $\Sigma\models\sigma$, it suffices to check whether $\Sigma^*\models\sigma^*$. Since $\Sigma^*$ and $\sigma^*$ can be computed in polynomial time from $\Sigma$ and $\sigma$, and FD+UIND implication is in polynomial time, also FD+UMI implication is in polynomial time.
\end{proof}
The above result clearly follows also from Theorem \ref{ptime_thm} to be proved in the next section, but as seen, it is already a simple corollary of the simulation result.
 
\section{Complexity of the implication problem for FDs+UMIs+UMDEs}\label{ComplexitySection}

In this section, we consider the implication problem for FDs+UMIs+UMDEs. We show that this implication problem also has a polynomial-time algorithm. Denote by $D$ the set of the variables appearing in $\Sigma\cup\{\sigma\}$, and suppose that $\Sigma$ is partitioned into the sets $\Sigma_{\text{FD}}$, $\Sigma_{\text{UMI}}$, and $\Sigma_{\text{UMDE}}$ that consist of the FDs, UMIs, and UMDEs from $\Sigma$, respectively.

First, note that if $\sigma$ is a UMI, it suffices to check whether $\sigma\in\cl(\Sigma_{\text{UMI}})$ as the inference rules for FDs and UMDEs do not produce new UMIs. Notice that $\Sigma_{\text{UMI}}$ can be viewed as an undirected graph $G=(D,\approx)$ such that each $x\approx y\in\Sigma_{\text{UMI}}$ means there is an undirected edge between $x$ and $y$. (We may assume that each vertex has a self-loop and all the edges are undirected, even though the corresponding UMIs might not be explicitly listed in $\Sigma_{\text{UMI}}$.) Suppose that $\sigma=x\approx y$. Then $\sigma\in\cl(\Sigma_{\text{UMI}})$ iff $y$ is reachable from $x$ in $G$. This can be checked in linear-time by using a breadth-first search algorithm.

Hence, we now treat the case where $\sigma$ is an FD or a UMDE. The idea behind the following part of the algorithm is again based on the algorithm \cite{CosmadakisKV90} for the implication problem for FDs and UINDs. In the FD+UMI case in Section \ref{simulation_ch}, we could use the algorithm for FD+UIND and just slightly modify the set $\Sigma$, but since now we are also dealing with UMDEs, so we have to make some modifications.

The implication problem for FDs is known to be linear-time computable by the \textit{Beeri-Bernstein algorithm} \cite{BeeriB79}. Given a set $\Delta\cup\{\dep(\bar{x},\bar{y})\}$ of FDs, the Beeri-Bernstein algorithm computes the set $\fdclosure(\bar{x},\Delta)=\{z\mid \Delta\models\dep(\bar{x},z)\}$. Then $\Delta\models\dep(\bar{x},\bar{y})$ holds iff $\var(\bar{y})\subseteq\fdclosure(\bar{x},\Delta)$. (Note that the tuple $\bar{x}$ can be empty.) The Beeri-Bernstein algorithm keeps two lists: $\fdlist(\bar{x})$ and $\attrlist(\bar{x})$. The set $\fdlist(\bar{x})$ consists of FDs and is updated in order to keep the algorithm in linear-time (for more details, see \cite{BeeriB79}). The set $\attrlist(\bar{x})$ lists the variables that are found to be functionally dependent on $\bar{x}$. The algorithm\footnote{This presentation of the algorithm is based on \cite{CosmadakisKV90}.} is as follows:
\begin{itemize}
\item[(i)] Initialization: 

Let $\fdlist(\bar{x})=\Delta$ and $\attrlist(\bar{x})=\var(\bar{x})$.
\item[(ii)] Repeat the following until no new variables are added to $\attrlist(\bar{x})$:
\begin{itemize}
\item[(a)]For all $z\in\attrlist(\bar{x})$, if there is an FD $\dep(u_0\dots u_k,\bar{u}')\in\fdlist(\bar{x})$ such that $z=u_i$ for some $i=0,\dots,k$, replace $\dep(u_0\dots u_k,\bar{u}')$ with $\dep(u_0\dots u_{i-1}u_{i+1}\dots u_k,\bar{u}')$.
\item[(b)] For each constant atom $\dep(\bar{u})\in\fdlist(\bar{x})$, add $\var(\bar{u})$ to $\attrlist(\bar{x})$.
\end{itemize} 
\item[(ii)] Return $\attrlist(\bar{x})$.
\end{itemize}
Note that the Beeri-Bernstein algorithm can be run incrementally, i.e., it can first examine a set of FDs and find their closure and then, if new FDs are added afterwards, it can continue computing the closure with respect to this larger set of FDs. \cite{CosmadakisKV90}

Let $x\in D$, and define $\desclist(x)=\{y\in D\mid\Sigma\models\dep(x,y)\}$. We will construct an algorithm, called $\mathsf{Closure}(x,\Sigma)$, that computes the set $\desclist(x)$ by utilizing the Beeri-Bernstein algorithm. If we view the variables $D$ and the atoms from $\cl(\Sigma)$ as a multigraph of Definition \ref{graph_def}, the set $\desclist(x)$ consists of the variables that are red descendants of $x$. Since FDs and UMIs interact only via UMDEs, it is sufficient to construct the multigraph only for FDs and UMDEs. We only consider the UMIs in the initialization step, as they imply UMDEs. The algorithm $\mathsf{Closure}(x,\Sigma)$ keeps and updates $\fdlist(y)$ and $\attrlist(y)$ for each $y\in D$. When the algorithm halts, $\attrlist(x)=\desclist(x)$.  The algorithm is as follows:
\begin{itemize}
\item[(i)] Initialization:
\begin{itemize}
\item[(a)] For each variable $y\in D$, $\attrlist(y)$ is initialized to $\{y\}$ and $\fdlist(y)$ to $\Sigma_{\text{FD}}$. We begin by running the Beeri-Bernstein algorithm for each $y\in D$. The algorithm adds to $\attrlist(y)$ all variables $z\in D$ that are functionally dependent on $y$ under $\Sigma_{\text{FD}}$.  Note that each $\fdlist(y)$ is also modified.
\item[(b)] We then construct a multigraph $G$ similarly as before, but the black edges are replaced with blue ones, i.e. the set of vertices is $D$, and there is a directed red edge from $y$ to $z$ iff $z\in\attrlist(y)$, and there is an undirected blue edge between $y$ and $z$ iff at least one of the following $y\approx z,z\approx y, y\approx^* z, z\approx^* y$ is in $\Sigma_{\text{UMI}}\cup\Sigma_{\text{UMDE}}$. 
\end{itemize}  

\item[(ii)] Repeat the following iteration until halting and then return $\attrlist(x)$:
\begin{itemize}
\item[(1)] Ignore the colors of the edges, and find the strongly connected components of $G$.
\item[(2)] For any $y$ and $z$ that are in the same strongly connected component, if there is a red directed edge from $y$ to $z$, check whether there is a red directed edge from $z$ to $y$ and a blue undirected edge between $y$ and $z$. If either or both are missing, add the missing edges.
\item[(3)] If no new red edges were added in step (2), then halt. Otherwise, add all the new red edges to $\fdlist(y)$.
\item[(4)] For each $y$, continue the Beeri-Bernstein algorithm.
\item[(5)]If for all $y$, no new variables were added to $\attrlist(y)$ in step (4), then halt. Otherwise, for each new $z\in\attrlist(y)$, add to $G$ a directed red edge from $y$ to $z$.
\end{itemize}

\end{itemize}

Let $n=|D|$. Since we can add at most $n^2$ new red edges to the graph, $|\fdlist(y)|\leq|\Sigma_{\text{FD}}|+O(n^2)$ for each $y\in D$. Thus the Beeri-Bernstein algorithm takes $O(|\Sigma_{\text{FD}}|+n^2)$ time for each $y\in D$. Because we run the algorithm for all $y\in D$, the total time is $O(n(n^2+|\Sigma_{\text{FD}}|))$. 

In each iteration, before running the Beeri-Bernstein algorithm (step (4)), we find the strongly connected components (step (1)) and check all pairs of vertices for red edges and add red and blue edges if needed (step (2)). The time required for finding the strongly connected components is linear in the number of vertices and edges \cite{Tarjan72}, so it takes at most $O(|D|+n^2+|\Sigma_{\text{UMI}}|+|\Sigma_{\text{UMDE}}|)$, i.e. $O(n^2)$ time. Going through the pairs of vertices in step (2) also takes $O(n^2)$ time. 
%Thus each iteration takes $O(n^3+n|\Sigma_{\text{FD}}|)$ time in total. 

An iteration proceeds to step (4) only if new red edges are added in step (2). In each iteration, if new red edges are added, they merge together two or more red cliques in the same strongly connected component. In the beginning we have at most $n$ such cliques, so the algorithm cannot take more than $O(n)$ iterations. Since the Beeri-Bernstein algorithm can operate incrementally, the total running time of $\mathsf{Closure}(x,\Sigma)$ is then $O(n^3+n|\Sigma_{\text{FD}}|)$.

If $\sigma=\dep(x,y)$, then $\Sigma\models\sigma$ iff $y\in\desclist(x)$. If $\sigma$ is a constant atom or a non-unary FD, in order to check whether $\Sigma\models\sigma$, it suffices to check whether $\Sigma_{\text{FD}}\cup\Delta\models\sigma$, where $\Delta=\{\dep(y,\bar{z})\mid y\in D, \var(\bar{z})=\desclist(y)\}$. This, again, can be done using the Beeri-Bernstein algorithm. Since the  order of the variables $\var(\bar{z})$ in the tuple $\bar{z}$ does not matter for each FD $\dep(y,\bar{z})$, we may assume that the set $\Delta$ contains only one FD for each $y\in D$. Thus this step only adds $O(|\Sigma_{\text{FD}}|+ n)$ extra time. (Note that constructing the set $\Delta$ does not affect the time-bound because $\mathsf{Closure}(x,\Sigma)$ actually computes $\desclist(y)$ for all $y\in D$.) 

If $\sigma=x\approx^* y$, we can use the blue (undirected) subgraph of $G$ to check whether $\Sigma\models\sigma$. After the algorithm $\mathsf{Closure}(x,\Sigma)$ has halted, we only need to check whether $y$ is reachable from $x$ in the subgraph, which is a linear-time problem, as noted before.

Therefore, in the worst case the above algorithm takes $O(n^3+n|\Sigma_{FD}|)$ time and we have the following theorem:
\begin{theorem}\label{ptime_thm}
The implication problem for functional dependencies, unary marginal identities, and unary marginal distribution equivalences is in polynomial time.
\end{theorem}
Note that the time-bound of the algorithm is the same as that of the algorithm for the implication problem for FDs+UINDs in \cite{CosmadakisKV90}.

\section{Conclusion}

We have presented a sound and complete infinite axiomatization for functional dependence, unary marginal identity, and unary marginal distribution equivalence, and shown that there is no finite axiomatization. The class of these dependencies was also shown to have Armstrong relations. We also showed that FD+UMI implication can be simulated with FD+UIND implication, and therefore it is in polynomial time. The implication problem for the full class of these dependencies was also shown to be in polynomial-time by constructing an algorithm. The following questions remain open:
\begin{itemize}
\item[(i)] Can we extend the axiomatization to nonunary marginal identity and marginal distribution equivalence?
\item[(ii)] Can we find an axiomatization for probabilistic independence and marginal identity (and marginal distribution equivalence)?
\item[(iii)] Can we find an axiomatization for (unary) functional dependence, probabilistic independence, (unary) marginal identity and marginal distribution equivalence? 
\end{itemize}

Concerning the first question, one should not rush to conjectures based on the results in the qualitative (non-probabilistic) case. Our results for functional dependencies and unary variants of marginal indentities and marginal distribution equivalences were obtained by methods that resemble the ones used in the case of functional dependence and unary inclusion dependence in \cite{CosmadakisKV90}.  The implication problem for functional and inclusion dependencies (where also nonunary inclusions are allowed) is undecidable for both finite and unrestricted implication \cite{chandra85}, so one might think that allowing nonunary marginal identities with functional dependencies makes the corresponding (finite) implication problem undecidable. However, it is not clear whether this is the case, as the reduction proof for the undecidability result does not directly generalize to the probabilistic case.

As for the second question, the implication problems for probabilistic independence and  marginal identity have already been studied separately: the problem for probabilistic independence has been axiomatized in \cite{geiger:1991}, and the problem for marginal identity (over finite probability distributions) has recently been axiomatized in \cite{hannula2021tractability}. Obtaining an axiomatization for both of these dependencies together is appealing, given how commonly IID assumption is used in probability theory and statistics. For probabilistic independence and \textit{unary} marginal identity and distribution equivalence, the interaction seems simple, but as in the case of functional dependence, adding \textit{nonunary} marginal identities and distribution equivalences introduces additional
complexity.

In the third question, we have both functional dependence and probabilistic independence together with marginal identity and marginal distribution equivalence. In the qualitative case, the decidability of the implication problem for independence atoms and (nonunary) functional dependencies is still open, but the axiomatization for independence atoms, \textit{unary} functional dependencies and \textit{unary} inclusion dependencies was introduced in \cite{hannula18}. This suggests that it might be useful to consider just unary functional dependencies with probabilistic independence and unary marginal identity and distribution equivalence. This already makes the interaction more interesting than in the case of only probabilistic independence and unary marginal identity and distribution equivalence because we would need the cycle rules.

In addition to probabilistic independence, we can consider probabilistic conditional independence $\pci{\bar{x}}{\bar{y}}{\bar{z}}$ which states that $\bar{y}$ and $\bar{z}$ are conditionally independent given $\bar{x}$. A probabilistic conditional independency of the form $\pci{\bar{x}}{\bar{y}}{\bar{y}}$\footnote{Sometimes the tuples of variables are assumed to be disjoint which disallows the probabilistic conditional independencies of this form. In the context of team semantics, such an assuption is usually not made.} is equivalent with the functional dependency $\dep(\bar{x},\bar{y})$, so functional dependence can be seen as a special case of  probabilistic conditional independence. It has been known since the 1990s that the implication problem for probabilistic conditional independence has no finite complete axiomatization \cite{studeny:1993}, and more recently that it is undecidable \cite{Li22}. Still, sound (but not complete) axiomatizations, the so-called graphoid and semigraphoid axioms, exist and are often used in practice. Every sound inference rule for probabilistic conditional independence that has at most two antecedents can be derived from the semigraphoid axioms \cite{studeny94}. It would be interesting to find axiomatizations for marginal identity and  some subclasses of probabilistic conditional independence. The results of this paper imply that if we consider a subclass that contains probabilistic conditional independencies of the form $\pci{\bar{x}}{\bar{y}}{\bar{y}}$, we effectively include functional dependencies and thus we will also have the cycle rules.

\begin{acks}
The author was supported by the Finnish Academy of Science and Letters (the Vilho, Yrj\"o and Kalle V\"ais\"al\"a Foundation) and by grant 345634 of the Academy of Finland. I would like to thank all the anonymous referees for valuable comments regarding both this and the previous conference version of the paper, and Miika Hannula and Juha Kontinen for useful discussions and advice.
\end{acks}

\bibliographystyle{ACM-Reference-Format}
\bibliography{biblio}

\end{document}